\documentclass[a4paper]{cas-sc}

\usepackage[authoryear]{natbib}

\usepackage{amsmath}
\usepackage{relsize}

\usepackage{graphicx}
\usepackage{xcolor}
\usepackage{mathrsfs}
\usepackage{amsthm}
\usepackage{hyperref}
\usepackage[normalem]{ulem}
\usepackage{booktabs}
\usepackage{adjustbox}
\usepackage{tikz}
\usetikzlibrary{scopes,patterns,intersections,calc,patterns.meta, shapes,arrows,positioning}
\tikzdeclarepattern{
  name=mylines,
  parameters={
      \pgfkeysvalueof{/pgf/pattern keys/size},
      \pgfkeysvalueof{/pgf/pattern keys/angle},
      \pgfkeysvalueof{/pgf/pattern keys/line width},
  },
  bounding box={
    (0,-0.5*\pgfkeysvalueof{/pgf/pattern keys/line width}) and
    (\pgfkeysvalueof{/pgf/pattern keys/size},
0.5*\pgfkeysvalueof{/pgf/pattern keys/line width})},
  tile size={(\pgfkeysvalueof{/pgf/pattern keys/size},
\pgfkeysvalueof{/pgf/pattern keys/size})},
  tile transformation={rotate=\pgfkeysvalueof{/pgf/pattern keys/angle}},
  defaults={
    size/.initial=5pt,
    angle/.initial=45,
    line width/.initial=.4pt,
  },
  code={
      \draw [line width=\pgfkeysvalueof{/pgf/pattern keys/line width}]
        (0,0) -- (\pgfkeysvalueof{/pgf/pattern keys/size},0);
  },
}

\usepackage{colortbl}
\usepackage{enumitem}

\newtheorem{definition}{Definition}
\newtheorem{proposition}{Proposition}
\newtheorem{theorem}{Theorem}

\newtheorem{example}{Example}
\newtheorem{rem}{Remark}
\newtheorem{assumption}{Assumption}
\newtheorem*{theorem*}{Theorem}
\newtheorem*{proposition*}{Proposition}
\newtheorem*{definition*}{Definition}
\newtheorem*{example*}{Example}

\usepackage[ruled,vlined,linesnumbered]{algorithm2e}
\SetKwComment{Comment}{\#}{}

\SetCommentSty{mycommfont}

\def\tsc#1{\csdef{#1}{\textsc{\lowercase{#1}}\xspace}}
\tsc{WGM}
\tsc{QE}
\tsc{EP}
\tsc{PMS}
\tsc{BEC}
\tsc{DE}

\shorttitle{Sliding-window signatures}
\shortauthors{}

\title [mode = title]{Sliding-Window Signatures as Features for Time Series Forecasting}                      

\author[inst1]{Nina Drobac} \corref{cor1}
\ead{nina.drobac@sorbonne-universite.fr}
\cortext[cor1]{Corresponding author}

\author[inst1,inst2]{Margaux Br\'eg\`ere}

\author[inst3]{Joseph de\ Vilmarest}

\author[inst1]{Olivier Wintenberger}

\affiliation[inst1]{organization={LPSM, Sorbonne Université},
addressline={4 Pl. Jussieu},
city={Paris},
country={France}}
 
\affiliation[inst2]{organization={EDF R\&D, Lab Paris-Saclay},
            country={France}}
 
\affiliation[inst3]{organization={Viking Conseil},
            addressline={26 rue de la Pie}, 
            city={Chartres},
            postcode={28000}, 
            country={France}}

\begin{document}

\begin{abstract}
The evolution of a time series is often influenced not only by its past values but also by external covariates, giving rise to forecasting with exogenous variables. Nonlinear and delayed effects between covariates and the target series make this task particularly challenging. We propose a forecasting framework based on Ridge regression with sliding-window signature features, combining the computational efficiency of (semi-)parametric models with the flexibility of nonparametric approaches. These features capture complex dependencies between the target and the exogenous variables without relying on prior assumptions about the form of these relationships or the influence of past observations. We establish theoretical guarantees in the discrete-time setting, including universality of approximation and stationarity for the sliding-window transform. To ease implementation, we introduce a novel sequential algorithm that exploits window overlap for efficient signature updates. The method is evaluated on both synthetic and real electricity demand data. Results show that signature features effectively encode temporal and nonlinear dependencies, yielding accurate forecasts competitive with those based on expert knowledge and recent neural network approaches.
\end{abstract}



\begin{keywords}
path signatures \sep time series forecasting \sep exogenous variables \sep electricity demand
\end{keywords}

\maketitle

\section{Introduction}

Time series forecasting is essential across all domains where anticipated scenarios inform critical decision-making, including energy management, healthcare systems, and financial markets. In many of these settings, the target series does not evolve solely according to its own past but is also influenced by external factors. For example, electricity demand depends strongly on meteorological conditions, making the inclusion of exogenous information through explanatory covariates, such as temperature, essential for accurate prediction. We therefore consider the general setting of time series forecasting with exogenous variables, where the goal is to forecast the target, also referred to as the endogenous variable, using observed exogenous time series that provide informative external signals. Accurate forecasting in this setting is inherently challenging, as it requires models capable of capturing nonlinear relationships and temporal dependencies between the target and the exogenous series.

This paper focuses in particular on forecasting electricity demand, a key challenge for maintaining the balance between production and consumption on the electric grid. As energy production is set according to predicted consumption, reliable forecasts are crucial in keeping the grid operating around the clock. Given the operational and economic significance of accurate demand prediction, a wide range of statistical and machine learning methods has been developed for this purpose; a comprehensive survey is provided in \citet{antoniadis2024statistical}.

Existing forecasting models can broadly be divided into two categories. On the one hand, (semi-)parametric models rely on expert-designed features and carefully tuned components. In particular, Generalized Additive Models (GAMs) \citep{wood2017generalized} are widely employed for electricity load forecasting \citep{pierrot2011short}, where splines or smoothing functions are applied to each exogenous variable in an additive manner. While these models are interpretable and computationally efficient, their performance depends heavily on domain knowledge and manual feature engineering. A typical example is the exponential smoothing of the temperature time series \citep{NEDELLEC2014375}, which models the delayed effect of temperature on electricity consumption due to the thermal inertia of buildings. This reliance on hand-crafted features makes the (semi-)parametric models less flexible in the presence of changing dynamics or the introduction of new explanatory variables.

On the other hand, nonparametric methods, most notably neural network (NN) architectures, learn representations directly from the data and are capable of capturing complex nonlinear relationships without expert intervention. Numerous successful applications of NN-based models to electricity load forecasting have recently been reported in the literature. For example, \citet{lstm_cnn} employ parallel Long Short-Term Memory (LSTM) and Convolutional Neural Network (CNN) architectures;  \citet{transformer_based} propose a transformer-based forecasting model; and \citet{keisler2024automateddeeplearningload} focus on automating the optimization of such architectures and their hyperparameters to improve predictive performance. Nevertheless, these approaches are typically computationally demanding, as they involve tuning millions of parameters, and they often lack interpretability. Moreover, as underlined in \citet{wang2024timexerempoweringtransformerstime}, many of the recent state-of-the-art architectures were designed for general multivariate forecasting, in which every series is simultaneously input and target. As such, they do not distinguish a target series from purely explanatory covariates, and are therefore not readily adapted to forecasting with exogenous variables.

In this work, we position ourselves between semi-parametric and nonparametric approaches, aiming to combine the computational efficiency and simplicity of the former with the flexibility and expressivity of the latter. Our objective is to propose a framework in which complex dependencies between the target and the exogenous variables are automatically captured, without relying on prior assumptions about the form of these relationships or the influence of past observations. At the same time, we seek a computationally efficient model that is well suited to resource-constrained environments and offers a more environmentally sustainable alternative to large-scale neural architectures. To this end, we leverage signatures as expressive time series features that eliminate the need for manual feature engineering while maintaining a simple and lightweight predictive model.

The signature is a collection of iterated integrals that provides a structured representation of sequential data. It was first introduced by \citet{Chen_1958} in the context of algebraic geometry, as a non-commutative formal power series for geometric paths. It was later rediscovered by \citet{Lyons1998}, who showed that augmenting a signal with its iterated integrals is key for extending classical differential equation theory to highly irregular driving signals. This insight laid the analytical foundations of rough path theory and led to a systematic study of signatures and their properties \citep{Lyons_2007,friz2010multidimensional, frizhairer2020}.

Path signatures have since been shown to capture key characteristics of multidimensional paths, including the temporal ordering of events and interactions between components, which makes them well suited for modeling sequential data. Their use as features is further supported by a universal approximation result: any continuous functional of a path can be approximated arbitrarily well by a linear functional of its signature. In this sense, the signature provides a way to transform a complex nonlinear learning problem into a linear one.
By representing data streams as continuous paths, one can compute their signatures and leverage them as an expressive non-parametric set of features in different machine learning frameworks ~\citep{chevyrev2026primersignaturemethodmachine,fermanian:tel-03507274}. This approach has proven effective across a wide range of machine learning tasks, with successful applications in fields such as medicine \citep{morrill_sepsis, cheng_cortical, kormilitzin2016applicationsignaturemethodpattern} and mathematical finance \citep{cuchieromoller2025, bayer2025pricingamericanoptionsrough}. Moreover, signatures have been increasingly integrated into deep learning frameworks, with established connections to recurrent neural networks ~\citep{fermanian2021framingrnnkernelmethod} and recent adaptations with transformer-based models for time-series tasks ~\citep{morenopino2025roughtransformerslightweightcontinuous}.

Continuing in this line of research, we propose a time series forecasting framework based on linear regression with signature features computed on sliding windows. While linear models with sliding-window signatures have already been employed for forecasting financial time series~\citep{cohen2023nowcasting, marketplace_forecast}, their theoretical properties have remained largely unexplored. Our work offers new insights and complements this literature by providing, to the best of our knowledge, the first formal proof that the sliding-window signature transformation preserves stationarity of the underlying series. The sliding-window construction is crucial both for the theoretical results, since classical signatures over the entire path do not preserve stationarity, and for practical application, as it allows the model to focus on the most recent history while discarding distant observations. A similar setting has recently been considered by \citet{jaber2025exponentiallyfadingmemorysignature}, who introduce the exponentially fading memory signature that gradually downweights remote past observations to maintain stationarity. In contrast, our approach relies on strict truncation via sliding windows, combined with a simple and efficient update algorithm that allows for a straightforward implementation.

The remainder of the paper is organised as follows. In Section \ref{sec:signatures}, we introduce the theoretical background on signatures and adapt it to the time series setting, interpreting sliding-window signatures as nonparametric features for forecasting. Section \ref{sec:theory} presents our main theoretical contributions. In particular, we establish a universal approximation property for linear functionals of sliding-window signatures in discrete time, providing theoretical justification for the proposed linear modeling approach. We also show that sliding-window signatures preserve stationarity when derived from a series with stationary increments. Section \ref{sec:implementation} covers key implementation aspects, introducing Ridge regression with signature features as a forecasting framework and discussing the associated hyperparameters. To facilitate its application, we furthermore propose a novel sequential algorithm that exploits the overlap between sliding windows to enable fast and efficient signature computation. Finally, Section \ref{sec:numerical_experiments} reports results on both synthetic and real-world data, demonstrating that our approach effectively leverages signature-based features to capture complex temporal dynamics and outperforms both linear baselines incorporating expert knowledge and most recent neural network models. Section \ref{sec:conclusion} concludes with a summary of our findings and final remarks.

\section{Signatures for Time Series Forecasting}
\label{sec:signatures}

In this section, we present signatures, along with necessary theoretical background, frame them as features and place them in a discrete-time forecasting setting. For a more rigorous mathematical treatment, with proofs and references to proofs, we refer to the Appendix \ref{appendix:theory}.

\subsection{Preliminaries}

We define a continuous path in $\mathbb{R}^d$ as any continuous mapping $x : [s, t] \longrightarrow \mathbb{R}^d$. We assume $d\geq2$.

\begin{definition}
    Let $x:[s,t] \longrightarrow \mathbb{R}^d$ be a continuous path.
    The 1-variation of $x$ is defined by
    \begin{equation*}
        \|x\|_{\mathrm{1}-\mathrm{var}}=\sup _{\left(t_0, \ldots, t_k\right) \in P} \sum_{i=1}^k\left\|x_{t_i}-x_{t_{i-1}}\right\|,
    \end{equation*}

    where $P=\{(t_0, \ldots, t_k) \mid k \geq 0,$ $s=t_0<\cdots<t_k=t\}$
    denotes the set of all finite partitions of $[s, t]$.
\end{definition}

Intuitively, $\|x\|_{\mathrm{1}-\mathrm{var}}$ can be seen as the length of the path $x$. 
From now on, we consider only paths of finite length, i.e. $\|x\|_{\mathrm{1}-\mathrm{var}} < \infty$.
We denote by $\otimes$ the tensor product and by $(\mathbb{R}^d)^{\otimes k}$ the $k$-th tensor power of $\mathbb{R}^d$ , with $(\mathbb{R}^d)^{\otimes 0} := \mathbb{R}$. 
We note that $(\mathbb{R}^d)^{\otimes k}$ is a Hilbert space of dimension $d^k$.

\begin{definition}
We denote by $\mathscr{T}$ the space of square-summable sequences of tensors of increasing order:
{\small
$$
    \mathscr{T}(\mathbb{R}^d) =
    \left\{ (a_k)_{k\ge0} \;\middle|\;
    a_k \in (\mathbb{R}^d)^{\otimes k},\;
    \sum_{k=0}^\infty \|a_k\|^2_{(\mathbb{R}^d)^{\otimes k}} < \infty
    \right\}.
$$
}
\end{definition}

We endow $\mathscr{T}(\mathbb{R}^d)$ with the scalar product
$
\langle \mathbf{a}, \mathbf{b}\rangle_{\mathscr{T}(\mathbb{R}^d)} = \sum_{k=0}^{\infty} \langle a_k, b_k\rangle_{(\mathbb{R}^d)^{\otimes k}} \,,
$
which induces the norm \\ $ \|\mathbf{a}\|_{\mathscr{T}\left(\mathbb{R}^d\right)} =  \sqrt{\sum_{k=0}^{\infty}\left\|a_k\right\|_{\left(\mathbb{R}^d\right)^{\otimes k}}^2} $.



\begin{proposition}
     $\left(\mathscr{T}\left(\mathbb{R}^d\right),\langle\cdot, \cdot\rangle_{\mathscr{T}\left(\mathbb{R}^d\right)}\right)$ is a separable Hilbert space. \label{prop:sig_Hilbert}
\end{proposition}

\subsection{Signatures}

\begin{definition}
    The signature of a finite-length path $x$ on $[s, t]$ is defined as an infinite tensor sequence:
    $$
        S\left( x_{[s,t]} \right)=\left(1, \ S^1\left( x_{[s,t]} \right), \ldots, \ S^k\left( x_{[s,t]} \right), \ \ldots \right),
    $$
    where the $k$-th element (called level) is given by
    $$
        {S^k \left( x_{[s,t]} \right)}=\idotsint \displaylimits_{s < u_1<\cdots<u_k < t}  \mathrm{d} x_{u_1} \otimes \cdots \otimes \mathrm{d} x_{u_k} \in {\left(\mathbb{R}^d\right)}^{\otimes k}.
    $$
\end{definition}

Each tensor $S^k(x)$ can be written in terms of its elements, referred to as signature coefficients, indexed by multi-indices $(i_1,\ldots,i_k) \in \{1,\ldots,d\}^k$:  
$$
    S^{\left(i_1, \ldots, i_k\right)} = \idotsint \displaylimits_{s < u_1<\cdots<u_k < t} \mathrm{d}x^{\left(i_1\right)}_{u_1} \cdots \mathrm{d}x^{\left(i_k\right)}_{u_k}.
$$

In practice, instead of dealing with infinite sequences, we only consider the levels up to (truncation) order $N$, defining the truncated signature as
$$
    S^{\leq N}\left( x_{[s,t]} \right) = \left(1, \ S^1\left( x_{[s,t]} \right), \ \ldots, \ S^N\left( x_{[s,t]} \right) \right).
$$
 Furthermore, we often ignore the tensor structure and view $S^{\leq N}$ as a collection of all signature coefficients with multi-index of length $k \leq N$, arranged in a vector of size $s_d(N)= \sum_{k=0}^{N} d^k = \left(d^{N+1}-1\right) /(d-1)$. We also note that, when the interval is not of importance, we use the abbreviation $S(x)$.

While the definition may look complex, we now illustrate that in the simple case of linear paths, there is no need for evaluating or numerically approximating the iterated integrals. Instead, we obtain a closed formula for signature coefficients which we could interpret as monomials of path increments.

\begin{example}[Signature of a linear path] \label{ex-lin-sig}
    Let $x:[s,t] \mapsto \mathbb{R}^d$, $u \to \frac{x_{t} - x_s}{t-s}(u-s) + x_{s}$ be a d-dimensional linear path. It follows by integration that $$S^{\left(i_1, \ldots, i_k\right)}\left(x_{[s,t]}\right) = \frac{1}{k!} \prod_{j=1}^k\left(x_t^{(i_j)}-x_s^{(i_j)}\right) $$ and $S^k\left( x_{[s,t]} \right)=\frac{1}{k!}\left(x_t-x_s\right)^{\otimes k}$. (Detailed calculation in Appendix \ref{appendix:theory_signatures}, Example \ref{ex:piecewise_2}.)
\end{example}

We now turn to key signature properties leveraged in our framework. It follows from their definition as iterated integrals that signatures are invariant to both translation and time reparametrization (see Proposition \ref{prop:invariances}). Consequently, they do not encode the path’s starting point nor the speed with which it was traversed. A standard approach to reintroduce the latter information is $\textit{time augmentation}$, i.e. adding time as a path component.

The following proposition and example will prove crucial for efficient signature calculation.

\begin{proposition}[Algebraic properties]
\label{prop:algebraic_properties}
Let $x:[s, t] \mapsto \mathbb{R}^d$ and $y:[t, u] \mapsto \mathbb{R}^d$ denote two paths of finite length.
\begin{enumerate}
    
    \item (\textbf{Chen's identity}) Let $x * y:[s, u] \mapsto \mathbb{R}^d$ be the concatenation of $x$ and $y$, meaning $(x * y)_v=x_v$ for $v \in[s, t]$ and $(x * y)_v=x_t+y_v-y_t$ for $v \in[t, u]$. Then
    $
    S(\left(x * y\right)_{[s,u]})=S\left( x_{[s,t]} \right) \otimes S(y_{[t,u]}) \label{chen}.
    $
    
    \item (\textbf{Time reversal}) We denote the time-reversal of $x$ as the path $\overleftarrow{x}:[s, t] \mapsto \mathbb{R}^d$ where $\overleftarrow{x}(u)=x_{s+t-u}$. Then
    $
    S\left( x_{[s,t]} \right) \otimes S\left( \overleftarrow{x}_{[s,t]} \right) = (1,0,0,\ldots).
    $
\end{enumerate}
 
\end{proposition}

Chen's identity allows us to extend the result from Example \ref{ex-lin-sig} to piecewise linear paths.

\begin{example}
\label{ex-piecewise}
    Let $x:[s,t] \rightarrow \mathbb{R}^d$ be a piecewise linear path and let $s=u_0<u_1<\cdots<u_k=t$ be a partition such that $x$ is linear on each $\left[ u_{j-1}, u_j \right]$. From Chen's identity we have:
    $$
    S\left( x_{[s,t]} \right)=S\left( x_{[s,u_1]} \right)\otimes S\left( x_{[u_1,u_2]} \right) \otimes \cdots \otimes S\left( x_{[u_{k-1},t]} \right),
    $$
    where every $S\left( x_{[u_{j-1},u_j]} \right)$ is given in Example \ref{ex-lin-sig}.
\end{example}

The next result (\citealp[Lemma 5.1]{lyons2014rough}) is key for switching from infinite to finite-dimensional objects.

\begin{proposition}
 Let $x:[s,t] \longrightarrow \mathbb{R}^d$ be a path of finite length. It holds that:
\begin{equation*}
\begin{aligned}
{\|{S^k \left( x_{[s,t]} \right)}\|}_{ {\left( \mathbb{R}^d \right)}^{\otimes k}} 
 & \leq \frac{\|x\|_{1 \text{-} var}^k}{k!}, \\
 \|S\left( x_{[s,t]} \right)\|_{\mathscr{T}\left(\mathbb{R}^d\right)} & \leq \exp \left(\|x\|_{1 \text{-} var}\right)<\infty.
\end{aligned}
\end{equation*}
\end{proposition}

The first equation can be intuitively interpreted as a bound on the information carried by each signature level: the factorial decay of the $k$-th level norm indicates that information  contained should decrease as the level increases. Consequently, for sufficiently large $N$,  $S^{\leq N}(x)$ provides a good approximation of $S(x)$, justifying truncation in practice. The second equation shows that signatures lie in the separable Hilbert space $\mathscr{T}(\mathbb{R}^d)$, a property we will exploit when working with probability distributions.
 
Finally, we present the main motivation for using truncated signatures as features in machine learning tasks.

\begin{theorem}[Universal approximation] \label{thm:uat}
    Let $K$ be a compact subset of the space of finite-length paths from $[s,t]$ to $\mathbb{R}^d$ and such that for any $x \in K,$ $x_s=a$ for some $a \in \mathbb{R}^d$ and $x$ has at least one monotone coordinate. Let $f: K \rightarrow \mathbb{R}$ be continuous. Then, for every $\varepsilon>0$, there exists $N \ge 1, \theta \in \mathbb{R}^{s_d(N)}$, such that, for any $x \in K$,
    \begin{equation*}
        \left\lvert f(x)- \theta^{\top} {S^{\leq N}\left(x_{[s,t]}\right)}\right\rvert \leq \varepsilon.
    \end{equation*}
\end{theorem}

This result follows directly from the Stone-Weierstrass theorem, showing that signature coefficients serve as the path analogue of monomials. Therefore, they can be leveraged to linearize the problem of learning a more complex (non-linear) function on a compact set of paths.

So far, our treatment of paths and signatures has been deterministic. Turning to the statistical perspective, from now on, we consider random paths
$X : [s, t] \to \mathbb{R}^d$. Indeed, for each $\omega$ in the sample space $\Omega$, the realisation  $X(\omega)$ is a path for which we can define the signature $S\left( x(\omega) \right)$ as before. In other words, $S(X)$ is a random variable taking values in a separable Hilbert space (Proposition \ref{prop:sig_Hilbert}).

\subsection{Signature Features for Time Series}
\label{sec:framework}

We now present the pipeline for using signatures as a feature extraction method for discrete-time time series forecasting. 

Let $\left( Y_t \right)_t$ denote the target time series and $\left( X_t \right)_t$ denote the covariate time series, with $Y_t \in \mathbb{R}$ and $X_t \in \mathbb{R}^d$.
We consider predictions for $Y_t$ of the form $f(X_{u \le t})$ for some continuous, measurable, possibly non-linear function $f$. 
In practice, when forecasting electricity demand, covariates from the distant past have little predictive value. 
Hence, we assume short-term dependence: the forecast of $Y_t$ depends only on a finite number of the most recent observations of $\left( X_t \right)_t$.
 
To apply the previously established results, discrete data must be first embedded into continuous time to yield a path. The choice of embedding can significantly impact the performance of signature-based methods~\citep{FERMANIAN_embedding}. We turn to linear interpolation, meaning that the path on the interval $\left[t, t+1\right]$ is defined by
$
u \longmapsto X_{t}+{(u-t)}\left(X_{t+1}-X_{t}\right).
$
By concatenating finitely many of these segments, we obtain piecewise linear paths. We note that these are finite-length paths and therefore fall within the theoretical framework established earlier.
Moreover, as illustrated in Example \ref{ex-piecewise}, they allow for a straightforward iterative computation, consisting in calculating the signatures on each linear segment as in Example \ref{ex-lin-sig} and concatenating them inductively through the tensor product (see Remark \ref{rem:piecewiselin_calc}). 
This procedure is efficiently implemented in the Python package \texttt{iisignature} ~\citep{iisignature}, used in this work.

In classical frameworks like that of ~\cite{fermanian2022functional}, signatures summarize the entire path. In time series, this corresponds to expanding windows, where all past points contribute equally. This is unsuitable for our application, where only recent observations matter. A more natural choice is a sliding window of size $w$, considering only data from $t-w$ to $t$. As will be shown in Section \ref{sec:theory}, sliding windows also help preserve stationarity, though they introduce $w$ as a hyperparameter, which we discuss in Section \ref{sec:numerical_experiments}.

Finally, we apply time augmentation within each sliding window, resulting in  $(0, X_{t-w}), (1,  X_{t-w+1} ), \ldots, ( w, X_t )$ at time $t$. This ensures that the considered paths will have at least one monotone coordinate, which is furthermore bounded.

Our method consists in the following steps.
We first fix a sliding window size $w$. At each time $t$, we consider $\left(0, X_{t-w} \right), \left(1, X_{t-w+1} \right), \ldots, \left( w, X_t \right)$. We linearly interpolate these points to obtain a piecewise linear path, denoted by $X_{[t-w,t]}$. Having constructed a continuous path, we can calculate its signature $S\left( X_{[t-w,t]} \right)$ and leverage the theory presented in Section \ref{sec:signatures}. In terms of application, we can now use the truncated signature $S^{\leq N} \left( X_{[t-w,t]} \right)$ as a feature set for $\left(X_t\right)_t$ on $[t-w,t]$. 
We summarize this as a pipeline in Figure \ref{figure:pipeline}. The final step of fitting a linear model is motivated by the universal approximation theorem, which we establish for our sliding-window, discrete-time setting in the following section. 

\begin{rem}
    \label{rem:time_augmentatiom}
  This approach yields the same signature values as if the time coordinate was added over the entire time horizon:  $\left(t-w, X_{t-w} \right), (t-w+1,$  $X_{t-w+1} ),\ldots, \left( t, X_t \right)$. This equivalence arises because signature computations depend only on increments of the path, and therefore on the relative time difference within each window.
  Consequently, when the window slides, only the new observation needs to be added without adjusting previous timestamps, which is crucial for establishing the sequential algorithm presented in Section \ref{subsec:algorithm}.
\end{rem}

\begin{figure*}[pos=htbp]
\begin{adjustbox}{width=\textwidth}
\centering
\begin{tikzpicture}[>=stealth, thick, scale=0.38]
    \definecolor{Green1}{HTML}{2D9253}
    \definecolor{Red1}{HTML}{B70E00}

    \node at (-15,0) {\begin{tabular}{c} Covariates \\ $X_{t-w}, \ldots, X_t$ \end{tabular}};
    \node at (-1.2,0) {\begin{tabular}{c} Path \\ $X_{[t-w,t]}$ \end{tabular}};
    \node at (10,0) { $S^{\leq N}\left( X_{[t-w,t]} \right)$ };
    \node at (23.8,0) { $\hat{Y}_{t}$ };

    \draw[->, line width=1pt, black]
        (-11.5,0) -- (-3,0)
        node[midway, above, black, font=\footnotesize\itshape] {Time augmentation}
        node[midway, below, black, font=\footnotesize\itshape] {Linear interpolation};

    \draw[->, line width=1pt, black]
        (0.6,0) -- (6.5,0)
        node[midway, above, black, font=\footnotesize\itshape] {Signature}
        node[midway, below, black, font=\footnotesize\itshape] {calculation};

    \draw[->, line width=1pt, Green1]
        (13.4,0) -- (22.8,0)
        node[midway, above, Green1, font=\footnotesize\itshape] {Fit linear function};

    \draw[->, line width=1pt, Red1]
        (-13,-1.2) to[out=-8,in=-172] (22.8,-1.2);

    \node[text=Red1, font=\footnotesize\itshape]
        at (4,-3.4) {Fit non-linear function};
\end{tikzpicture}
\end{adjustbox}
\caption{Method Pipeline: leveraging sliding-window signature features to linearize the forecasting task.} \label{figure:pipeline} 
\end{figure*}
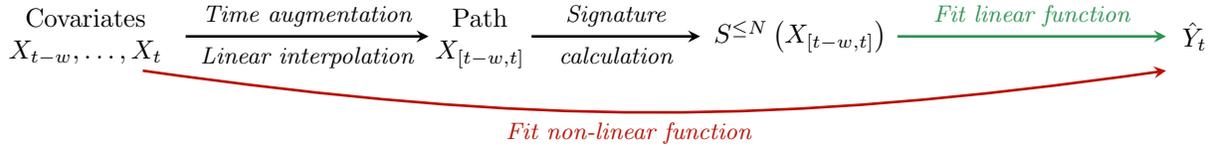

\section{Theoretical contributions}\label{sec:theory}
In this section, we identify conditions under which the universal approximation theorem holds and examine the important property of stationarity within our framework.

We first introduce two assumptions on the  discrete-time time series of covariates ${(X_t)}_t$.

\begin{assumption}[Uniform boundedness]\label{assumption:bounded}
    Suppose that $|X^{(i)}_t| \leq M$ \textit{a.s.} for some $M > 0$ and all $1 \leq i \leq d$.
\end{assumption}

\begin{assumption}[Stationarity of increments]\label{assumption:stationary}
    Suppose that the increments $(X_{t+1}-X_{t})_{t \ge 1}$ are strictly stationary.
\end{assumption}

Assuming that the time series is uniformly bounded is not restrictive in practice, as real-world data such as temperature naturally lie within physical limits. Moreover, a uniform bound $M$ can always be ensured by standardizing the data. Similarly, stationarity of increments is a milder assumption than stationarity of the full series; for example, while temperature exhibits seasonality, its half-hour increments can reasonably be considered stationary (for an empirical justification, see Appendix \ref{appendix:experiments}).

\subsection{Universal Approximation Theorem}

We treat the paths on sliding windows as functions ${X}_{[t-w, t]}(\omega) : [0,w] \to \mathbb{R}^{d+1}$, which share a common time domain, allowing us to consider the set of such paths and functions over this set.

\begin{theorem}[Universal approximation on sliding windows]
    Let ${(X_t)}_t$ be a $d$-dimensional time series such that Assumption \ref{assumption:bounded} holds. Let $w \in \mathbb{N}^\star$ be a fixed window size and $f : C\left([0,w],\mathbb{R}^{d+1}\right) \to \mathbb{R}$ a continuous function given the uniform topology. 
    It holds that for every $\varepsilon>0$, there exists $N \in \mathbb{N}, \theta \in \mathbb{R}^{s_{d+1}(N)}$, such that
    \begin{equation*}
        \sup_{t}\left\lvert f\left({X}_{[t-w, t]} - (0,X_{t-w})\right)- \theta^{\top} S^{\leq N} \left({X}_{[t-w, t]}\right) \right\rvert \leq \varepsilon.
    \end{equation*}
    \label{thm:uat-signatures}
\end{theorem}

\textit{Sketch of proof.} We first introduce paths on windows translated by the starting point, ${X}_{[t-w, t]} - \left( 0, X_{t-w} \right)$. We consider the set of all realisations of these random paths on all sliding windows. These paths have a monotone coordinate (time augmentation), they all begin at the same point (translation), are uniformly bounded (Assumption \ref{assumption:bounded}) and they all are Lipschitz with the same constant (Assumption \ref{assumption:bounded} and piecewise linearity). By constructing a superset of paths with the same properties that is compact in the uniform topology, we can apply Theorem \ref{thm:uat} to it, completing the proof. For more detail, see Appendix \ref{appendix:theory_uat}.

\begin{rem}
    \label{rem:subtraction}
  Assuming all paths start at the same point is vital, since signatures cannot encode the path's starting value due to translation invariance. This is commonly addressed in practice through basepoint augmentation, prepending a zero to each path ~\citep{morrill2021generalisedsignaturemethodmultivariate}. As it was not suitable for our sliding-window framework, we instead subtract the initial value from each window. The truncated signature remains unchanged under this transformation, $S^{\le N}\bigl(X_{[t-w,t]} - (0, X_{t-w})\bigr) = S^{\le N}(X_{[t-w,t]})$, again due to translation invariance. This is key for efficient sequential computation, as there is no need to modify the path at each step. The subtraction has mainly a theoretical implication: we consider the function $f$ as acting on translated paths, though in practice the starting point still carries predictive information. We address this explicitly in Section~\ref{sec:implementation}.
\end{rem}

The main takeaway from this theorem is that we can effectively linearize the task of learning the mapping $f$ by approximating it with a linear functional on the truncated signature. The result allows for an approximation of a wide range of complicated functions $f$ (such as GAMs or conditional expectation), without requiring assumptions on their specific form, but given that they are continuous with respect to the input.

\subsection{Stationarity of Signatures}

Computing the signature of $(X_t)_t$ over a sliding window produces a new signature series ${\left( S(X_{[t-w,t]}) \right)}_t$. Since these signatures form the features of our forecasting model, we study their stationarity to better characterize the resulting prediction process $(\widehat{Y}_t)_t$.

\begin{rem} \label{rem-stat_def}
    Although signatures are defined as infinite sequences, the usual definition of strict stationarity (via equality in distribution) still applies, since they lie in the separable Hilbert space $\mathscr{T}(\mathbb{R}^d)$, where finite-dimensional marginal distributions determine the law of the process~\citep[see, e.g. Chapter 1 of][]{Bosq2000LinearPI}.

\end{rem}

\begin{theorem}
Let ${(X_t)}_{t}$ be a discrete $d$-dimensional time series such that Assumption \ref{assumption:stationary} holds. Then, for any truncation order $N \in \mathbb{N}$, the time series ${\left({S^{\leq N}\left(X_{[t-w,t]}\right)}\right)}_t$ is strictly stationary. Furthermore, the time series ${\left({S\left(X_{[t-w,t]}\right)}\right)}_t$ is strictly stationary.
\label{thm:stationarity_sig}
\end{theorem}
    
\textit{Sketch of proof.} The main argument is that signature coefficients, and consequently truncated signatures, are deterministic, measurable functions of the series' increments. Applying any measurable, deterministic function to two random vectors with the same distribution preserves their equality in distribution. Therefore, given Assumption \ref{assumption:stationary}, the stationarity of truncated signatures holds. The conclusion for signatures follows immediately from the previous result and Remark \ref{rem-stat_def}. For full proof, see Appendix \ref{appendix:theory_stat}.

\begin{rem}
    Let us now suppose that Assumptions \ref{assumption:bounded}, \ref{assumption:stationary} hold for our time series of covariates $(X_t)_t$. Combining the two theoretical results, we can now conclude that the prediction process $( \widehat{Y}_t)_t$, with $\widehat{Y_t} = \theta^{\top} {S^{\leq N}\left({X}_{[t-w, t]} \right)}$, is also strictly stationary as a linear transformation of a stationary process. This has important implications for choosing a model. \label{rem:stationarity}
\end{rem}

\section{Implementation}\label{sec:implementation}

\subsection{Model}
\label{subsec:framework}

The preceding theoretical considerations yield two key implications for the specification of our model. First, our forecast is a function of the translated paths $ X_{[t-w]} - \left(0,X_{t-w}\right) $. Second, in accordance with Remark \ref{rem:stationarity}, under the Assumption \ref{assumption:stationary}, our forecast is stationary. It is therefore natural to apply our model on a stationary target variable. These requirements are, however, restrictive and fail to hold in the context of electricity demand. To address this, we introduce an additional hyperparameter: a delay $D$ in the target variable. The delay $D$ is selected to be sufficiently large to ensure that at time $t$ the true value of $Y$ at time $t-D$ is observable, but small enough that the increment process $\Delta_D Y_t := Y_t - Y_{t-D}$ can be treated as stationary. 
Furthermore, it is far more realistic that the increments $\Delta_D Y_t$, rather than the raw values $Y_t$, depend on the translated paths (see Remark \ref{rem:subtraction}). 

We now present our pipeline applied to $\Delta_D Y_t$. The training procedure consists in fitting a Ridge regression:
\begin{align*}
\hat\theta \in \arg\min\limits_\theta \sum\limits_{t=1}^{t_{train}} \left(\Delta_D Y_t - {\theta}^{\top} S^{\leq N} \left( X_{[t-w,t]} \right)\right)^2 + \lambda \|\theta\|_2^2 \, ,
\end{align*}
where the regularization constant $\lambda > 0$ is chosen on the validation set. Our final forecast of the target for any $t> w$ is then $\widehat{Y}_t = Y_{t-D} + \widehat{\Delta_D Y_t}$, with $\widehat{\Delta_D Y_t} = {\hat{\theta}}^{\top} S^{\leq N} \left( X_{[t-w,t]} \right)$. See Appendix \ref{appendix:implementation} for the full algorithm in pseudocode.


Several challenges arise when implementing this model. We begin by addressing the choice of hyperparameters: namely $D, w$ and $N$. The delay $D$ depends on the data and is chosen so that the assumptions of stationarity of $\left( \Delta_D Y_t \right)_t$ and the availability of $Y_{t-D}$ at time $t$ hold. When selecting the truncation order $N$, we highlight that the size of the truncated signature $s_d(N)$ grows polynomially with the dimension of the covariate series $d$, but exponentially with $N$. Consequently, large values of $N$ (above 10) are typically infeasible. In practice, $N$ is either chosen on a validation set or fixed to a moderate value.
The window size $w$ can likewise be selected on the validation set, noting that only $w \geq D$ are considered so that $\Delta_D Y_t$ may be forecasted as a function of $X_{[t-w]} - \left(0,X_{t-w}\right)$.

The use of a Ridge regression is motivated by the ill-conditioning of the design matrix. High correlations among signature components, amplified by overlapping sliding windows, lead to unstable coefficient estimates. We note that  LASSO and Ridge regression are common when working with signature features ~\citep{guo2024consistencysignatureusinglasso, fermanian2022functional}; we adopt Ridge, as it performs better in presence of collinearity. Additionally, we omit all sliding-window signature coefficients that depend solely on time. The $k$-th such coefficient is given by a constant $S^{(1,\ldots,1)}(x_{[t-w,t]}) = \frac{w^k}{k!}$, contributing only redundant intercept terms and worsening the collinearity. We thus redefine $S^{\leq N}$ to exclude these coefficients.

Finally, the truncated signature’s rapidly growing size (in $N$ and $d$) and the tensor-multiplication involved in its computation make calculating signatures for large and numerous windows computationally expensive. To address this, we propose a novel algorithm that incrementally updates the signature as the window slides, rather than recomputing it from scratch.

\subsection{Efficient Signature Computation}
\label{subsec:algorithm}

The key observation is that in our pipeline, there is a large overlap in paths captured within two consecutive sliding windows. 
To exploit this, we use the algebraic properties of signatures from Proposition \ref{prop:algebraic_properties}. Recall that the time-reversal property allows us to remove the contribution of $X_{t-w}$ by computing the signature of the linear section from $\left(t-w+1, X_{t-w+1}\right)$ back to $\left(t-w, X_{t-w}\right)$ (reversing time) and concatenating it on the left of the path using Chen’s identity. Similarly, to incorporate the new data point $X_{t+1}$, it suffices to compute the signature of the linear section between $\left(t, X_t\right)$ and $\left(t+1, X_{t+1}\right)$ and concatenate it on the right, again via Chen’s identity. These sequential updates are possible because the path remains unchanged as the window slides, which is still consistent with the theoretical insights, as highlighted in Remarks~\ref{rem:time_augmentatiom} and~\ref{rem:subtraction}.

With this procedure, summarised in Algorithm~\ref{algo:update_sig} and implemented in Python, the number of tensor products needed to update the signature on a new window drops from $w$ to just two. Consequently, for the naive approach of recomputing the signature on each window from scratch, the number of operations and runtime grow linearly with the window length, whereas they remain independent of the window size for our sequential algorithm (up to the  computation of the very first signature). Figure~\ref{fig:algo_comparison}  illustrates this: the runtime of our algorithm remains essentially constant as the window size grows, while that of the naive algorithm increases linearly with the window size. Even though recomputing is somewhat faster with only a few points in the window, the true gain of using the sequential algorithm is best seen for large window sizes. For example, in the experimental setup of Figure~\ref{fig:algo_comparison} and for a window size of $w=1000$, our algorithm is 65 times faster. 
Furthermore, the sequential algorithm preserves numerical precision as its signatures match the recomputed ones up to a small relative error. Over a large number of sliding steps, this error can accumulate, but it is easily kept in check by periodically recomputing the signature from scratch.

In conclusion, the sequential algorithm keeps signature computation tractable even on a very fine time grid with a large window size.

\begin{algorithm}
    \SetKwInOut{Input}{Input}
    \SetKwInOut{Output}{Output}
    \Input{data $\left\{ x_t \mid t \geq 0 \right\}$, window size $w$, truncation order $N$}
    \Output{signatures $S^{\leq N}\!\left(x_{[t-w,t]}\right)$ for all $t \geq w$}

    Compute initial signature $S \gets S^{\leq N}\!\left(x_{[0,w]}\right)$ \;

    \For{$t > w$}{
        
        Compute signature of the oldest segment $S_{\text{old}} \gets S^{\leq N} \left(x_{[t-w-1,t-w]}\right)$ \;
        Update $S$ by overwriting the oldest segment: $S \gets S_{\text{old}} \otimes S$ \tcp*{time-reversal update}

        Compute signature of the new segment $S_{\text{new}} \gets S^{\leq N}\!\left(x_{[t-1,t]}\right)$ \;
        Update $S$ by adding the new segment $S \gets S \otimes S_{\text{new}}$ \tcp*{Chen’s identity}
    }
    \caption{Signature update on sliding windows}
    \label{algo:update_sig}
\end{algorithm}

\begin{figure}[pos=htbp]
    \centering
    \includegraphics[width=0.55\linewidth]{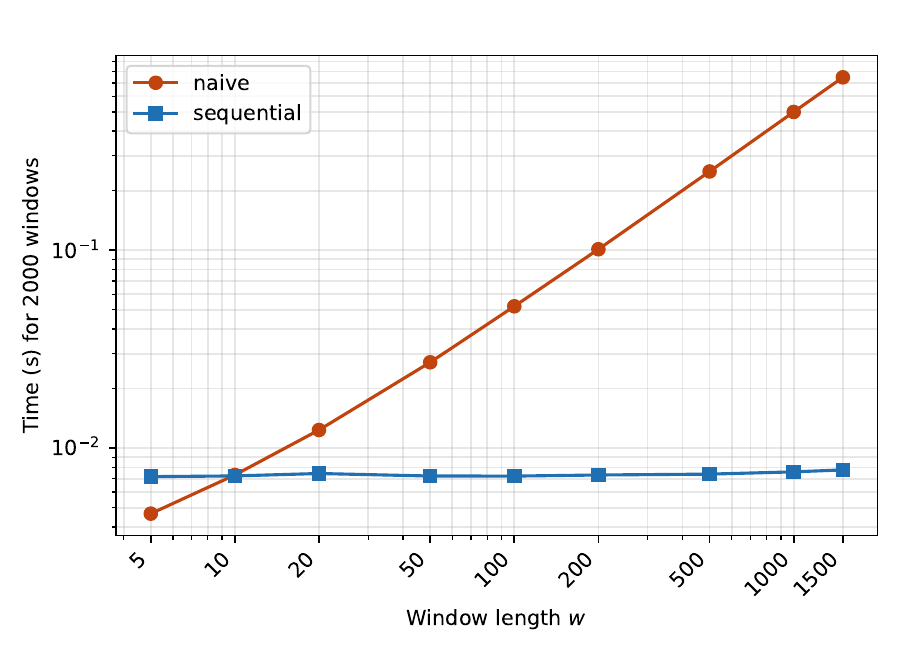}
    \caption{Median runtime (over 20 walks) versus window length $w$ for the naive
    and sequential sliding-window signature algorithms, on $3$-dimensional Gaussian
    random walks with signature truncation order $N=4$, each evaluated over $2000$
    windows. The maximum absolute difference across all signature components, windows and walks is of order $10^{-7} $and the maximum relative difference is of order $10^{-11}$.}
    \label{fig:algo_comparison}
\end{figure}

\section{Numerical experiments}\label{sec:numerical_experiments}

In this section, we evaluate our forecasting framework based on sliding-window signatures using both synthetic and real-world data. Our focus is on forecasting electricity demand in France, which exhibits a strong nonlinear dependence on recent past values of covariates see, e.g., Section 1.2.4 of \cite{antoniadis2024statistical}. This time series also has strong autocorrelative patterns; therefore, it is useful to include lagged observations as features (see, among others, \cite{Vilmarest2022state}). Our experiments are aimed at examining the effectiveness of signatures in capturing such non-linear and past dependencies of the target on the covariates. The application to synthetic data provides a proof of concept and enables a systematic analysis of the model’s hyperparameters under varying conditions. We further validate our approach on a real-world open-source dataset, comparing it against both linear models and recent NN–based approaches. Both the code and the data are made openly available; see Appendix \ref{appendix:experiments}.

\subsection{Dataset}

We gather half-hourly electrical demand data from Eco2mix\footnote{https://www.rte-france.com/eco2mix} open-source dataset published by RTE, France's electricity transmission system operator.
We combine it with temperature observation data from Météo France\footnote{https://donneespubliques.meteofrance.fr}. 
To obtain a time series of temperatures at the national level, we average the observations across French weather stations.
We then linearly interpolate the data to get a unified data set comprising 48 observations per day, covering the period from January 1, 2012, to December 30, 2015. 
This stable period was chosen to avoid the significant fluctuations in electricity demand caused by the COVID-19 pandemic and the 2022 energy crisis.
In what follows, we denote by $Y_t$ the electricity demand and by $T_t$ the temperature at any time $t=1, 2, \dots$.
We apply our method for a single feature: $T_t$, and underline that in practice, it is replaced with temperature forecasts. Furthermore, for any smoothing parameter $\alpha \in [0,1]$, we can define exponentially smoothed temperature as
\begin{equation*}
    \left\{
    \begin{array}{ll}
        \overline{T}_1^\alpha = T_1 \\
      \overline{T}_t^\alpha = (1-\alpha)\overline{T}_{t-1}^\alpha + \alpha T_t^\alpha \,, \,\textrm{ for any } t \geq 2  \,. 
    \end{array}
\right.
\end{equation*}
As mentioned in the introduction, smoothed temperature values are often used to model thermal inertia and were shown to enhance the quality of forecasts.
These variables stem from expert knowledge and will be useful for generating synthetic data and comparing the signature method to relevant baselines.

\subsection{Experimental Setup}  

In what follows, we represent the data as a two-dimensional path consisting of temperature augmented with rescaled time, i.e., $(t/w, T_t)$, where $w$ denotes the window size. Time is rescaled so that its increments, and therefore the scale of signature coefficients, remain comparable across different choices of $w$. The selection of window size is data-driven: since half-hourly demand strongly depends on the time of day, we preserve this structure by considering window sizes in units of full days. For the delay, we fix $D=2 \text{ days}$ in the case of synthetic data, reflecting the availability of observed demand, and $D=7 \text{ days}$ for the real data set, to match the periodicity of electricity consumption.
We use the years 2012 and 2013 to train the model, the year 2014 as a validation set on which we fix the hyperparameters (namely the Ridge regularization constant $\lambda$) and the year 2015 as a test set. 

Given a window size $w$ and a signature truncation order $N$, we predict the target time series using the model described in Subsection~\ref{subsec:framework}, which we denote by $\textit{RidgeSig(X)}$. We evaluate our method and benchmark models using the following standard error metrics:
Root Mean Squared Error (RMSE) and Mean Absolute Percentage Error (MAPE) on the test set.
We recall that for a target values $Y_1,\dots Y_T$ and the associated forecasts $\hat{Y}_1,\dots \hat{Y}_T$, we have:
$$
\mathrm{RMSE} = \frac{1}{T}\sum_{t=1}^T \big( \hat{Y}_t-Y_t \big)^2, \quad \mathrm{MAPE} = \frac{100}{T}\sum_{t=1}^T\frac{|\hat{Y}_t-Y_t|}{|Y_t|}.
$$

\subsection{Proof of Concept on Synthetic Data}
\label{subsec:expe:syntheticdata}

\paragraph{Data generation} 
For a smoothing parameter $\alpha$ and any time step $t$, we generate the electricity demand time series $\big(\widetilde{Y}^\alpha_t\big)_t$ as:
$$ \widetilde{Y}^\alpha_t = \theta_1 + \theta_2 \overline{T}_t^\alpha +  \theta_3 \big(\overline{T}_t^\alpha\big)^2 + \varepsilon_t, \quad \varepsilon_t \overset{\mathrm{i.i.d}}{\sim} \mathcal{N}(0,\sigma^2)\,. $$
Therefore, it depends non-linearly on recent and present temperatures, with $\alpha$ controlling the impact of past temperature values - the smaller $\alpha$, the longer the memory and the influence of past values.
We refer to Appendix ~\ref{appendix:experiments}
for details on fitting the model parameters ($\theta_1$, $\theta_2$, $\theta_3$ and $\sigma$) and synthetic data fidelity with real demand data.

\paragraph{Capturing temporal and non-linear dependencies}
To assess the effectiveness of sliding-window signatures in capturing temporal and non-linear dependencies of the target on covariates, we compare our Ridge regression model based on signature features with linear regression benchmarks that use covariates incorporating insider knowledge of the data-generating process. These include the smoothed temperature $\overline{T}_t^\alpha$ with the correct smoothing parameter, as well as using temperature $T_t$ and its square $T_t^2$ to account for the imposed quadratic relationship. We denote these linear regression benchmarks on various feature sets $X^1, X^2, \dots$ by $\mathrm{LR}(X^1, X^2, \dots)$.

For this experiment, we fix the smoothing parameter at $\alpha = 0.005$, which optimizes the performance of $\mathrm{LR}(\overline{T}_t^\alpha, (\overline{T}_t^\alpha)^2)$ and yields the closest fit to the real data. The signature truncation order is set to $N = 4$. For readability, we report results for the \textit{RidgeSig(X)} model for the best-performing window size $w = 9$ days. Table~\ref{table:experiment1} presents test-set results relative to the oracle model $\mathrm{LR}\big(\overline{T}_t^\alpha, (\overline{T}_t^\alpha)^2\big)$, demonstrating that \textit{RidgeSig(X)} clearly outperforms all other benchmarks. Notably, this advantage proved consistent across the entire range of window sizes considered (2–32 days).

Overall, these findings indicate that the signature-based approach captures both past dependencies, as it outperforms $\mathrm{LR}(\overline{T}_t^\alpha)$, and nonlinear effects, since it outperforms $\mathrm{LR}\big(\overline{T}_t^\alpha, (\overline{T}_t^\alpha)^2\big)$.

\begin{table}[pos=htbp]
\centering
\resizebox{0.7\textwidth}{!}{%
\begin{tabular}{lccccc}
\toprule
\textbf{Model} & $\mathrm{LR}(T)$ & $\mathrm{LR}(T, T^2)$ & $\mathrm{LR}(\overline{T}^{\alpha})$ & RidgeSig(X)& $\mathrm{LR}\big(\overline{T}_t^\alpha, (\overline{T}_t^\alpha)^2\big)$  \\
\midrule
\midrule
\textbf{RMSE (MW)} & 5\,435 & 4\,729 & 3\,079 & \cellcolor{lightgray} 1\,637 &  \textbf{995} \\
\textbf{MAPE (\%)} & 8.3 & 6.4 & 4.9 &  \cellcolor{lightgray} 2.5 &  \textbf{1.5}  \\ 
\bottomrule
\end{tabular}%
}
\caption{RMSE (MW) and MAPE (\%) on test set for synthetic data generated with $\alpha=0.005$.}
\label{table:experiment1}
\end{table}

\paragraph{Optimal sliding window size}
We now investigate how the optimal window size is affected by the strength of past dependencies in the data.
Figure~\ref{fig:combined} shows the RMSE of our approach with truncation order fixed at $N=5$, plotted for window lengths ranging from 2 to 32 days, applied to synthetic data generated by two different smoothing parameters $\alpha = 0.005$ and $\alpha = 0.05$.
This confirms the intuition that with $\alpha$ decreasing, the optimal window size increases.
Capturing a longer memory of the data-generating process requires a larger window so that the signature features can encode a longer history of the covariates.

\begin{figure}[pos=htbp]
    \centering
    \begin{minipage}{0.48\linewidth}
        \includegraphics[width=\linewidth]{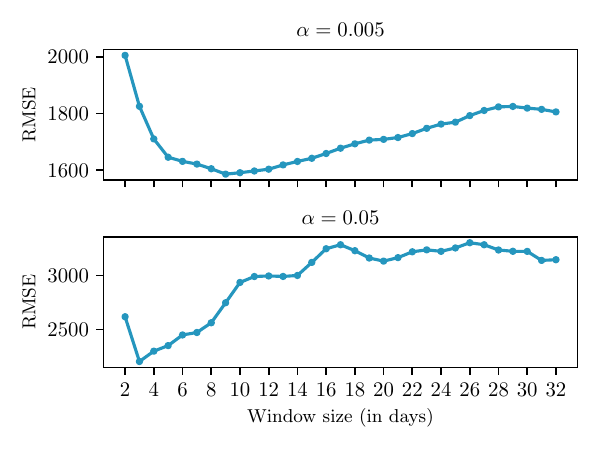}
    \end{minipage}
    \hfill
    \begin{minipage}{0.48\linewidth}
        \includegraphics[width=\linewidth]{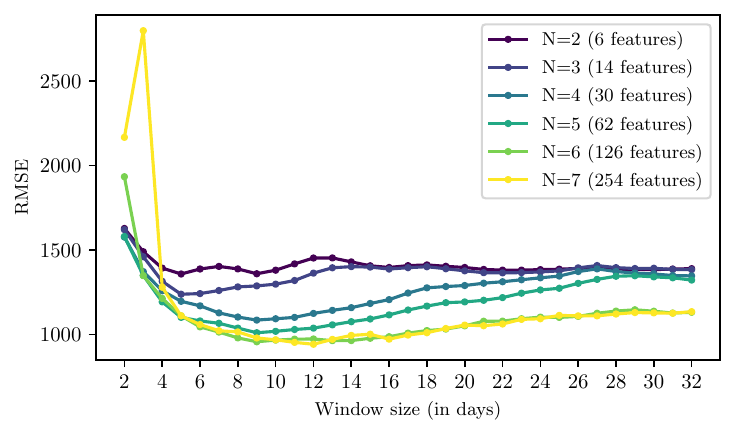}
    \end{minipage}
    \caption{Left: RMSE (MW) on the test set as a function of sliding window size for synthetic data ($\alpha=0.005$, optimal window = 9 days) and ($\alpha=0.05$, optimal window = 3 days). Right: RMSE (MW) on the test set for synthetic data ($\alpha$=0.005) as a function of sliding window size and different truncation orders.}
    \label{fig:combined}
\end{figure}

\begin{rem} 
 We can show by induction that $\overline{T}_t^\alpha = \sum_{s=0}^{t-1} \alpha (1-\alpha)^{s}T_{t-s} + (1- \alpha)^{t} T_1$, so the ratio between the proportion of $T_{t-s}$  and the one of $T_{t}$ in the smoothing equals $(1-\alpha)^s$. 
  After $3$ days, this ratio still equals $0.48$ for $\alpha = 0.005$, while it is already $6e^{-4}$ for $\alpha = 0.05$. 
  This suggests that the temperature from three days ago continues to have a significant impact on the data generated by the first model, whereas it is already being forgotten by the second model.
  Figure~\ref{fig:combined} perfectly illustrates this phenomenon. 
\end{rem}

\paragraph{Optimal signature truncation order}
We investigate the impact of the signature truncation order by running our algorithm for $N=2$ to $N=7$ and window sizes from $2$ to $32$ days.
As shown in Figure \ref{fig:combined}, an order that is too small makes it difficult to detect temporal and non-linear dependencies, which degrades performance.
Therefore, $N$ must be large enough.
On the other hand, we  highlight that the larger the value of N, the longer the computational time and the more unstable the results.
This suggests that the truncated order should be between 4 and 6. We also note that the large RMSE for small window sizes and high truncation orders is most likely caused by the fact that the dimension of the representation is larger than the number of the time steps included  in the window, giving rise to the problem of non-identifiability.

\subsection{Electricity Demand Forecasting}
\label{subsec:expe:realdata}
We now apply our procedure to real electricity demand data. As before, our goal is to forecast electricity demand using temperature as the exogenous variable. In this case, however, we expand the set of benchmark models and adapt \textit{RidgeSig} by explicitly accounting for the pronounced weekly cycle in the data and incorporating past values of the target.

\paragraph{Benchmarks}

As highlighted in the introduction, we propose linear models with signature features as a practical forecasting approach that remains as lightweight and computationally efficient as classical (semi-)parametric methods, while relying on a fully data-driven feature construction. To this end, we wish to show that our model combines the parsimony of traditional approaches with the data-driven adaptability that is one of the key strengths of non-parametric methods, namely those based on NNs. 

We therefore benchmark the performance of our \textit{RidgeSig} model against both: (i) (semi-)parametric models using hand-crafted features known to be effective for electricity demand forecasting and (ii) recent state-of-the-art NN architectures. In the first category, we consider a linear regression model incorporating (smoothed) temperature, its square and lagged demand values, as well as a GAM on temperature, smoothed temperature ($\alpha=0.01$) and consumption from seven and two days ago. Among NN approaches, we focus on those that allow the inclusion of exogenous variables. Specifically, we evaluate three distinct architectures: \textit{TiDE} ~\citep{das2024longtermforecastingtidetimeseries}, a multi-layer perceptron–based model designed for long-horizon prediction; \textit{TimeXer} ~\citep{wang2024timexerempoweringtransformerstime}, a transformer-based architecture combining patch-wise self-attention to capture autoregressive dynamics with variate-wise cross-attention to model exogenous dependencies; and \textit{N-BEATS} ~\citep{oreshkin2019n}, a deep fully connected residual network employing backward–forward residual stacks for interpretable time series decomposition.

Several key aspects of the NN-based models must be addressed to ensure a fair comparison with signature-based models. First, NN models always rely on past values of the target to make predictions, whereas in the signature-based model, including past target values is optional, as discussed below. Second, NNs produce predictions over a forecasting horizon. To ensure comparability, we fix the horizon to two days for all NN models and retain only the final point as the point forecast when computing error metrics. The forecasting horizon is set to reflect the availability of historical consumption data, which is typically accessible only up to two days prior. Finally, the considered implementations of \textit{TimeXer} and \textit{N-BEATS} do not support the use of known future values of exogenous variables, i.e. using $X_t$ to forecast $Y_t$. This limitation is particularly relevant for electricity load forecasting, where temperature forecasts carry substantial predictive value. Table~\ref{table:model_inputs} summarizes the information available at time $t$ to each model, including the two models introduced in the next paragraph, for producing the forecast $\hat{Y}_{t+2\text{days}}$. We refer to Appendix \ref{appendix:implementation_nn} for precise details on the implementation. 

\paragraph{RidgeSig models}
We change the delay parameter to $D=7$ to capture the pronounced weekly seasonality exhibited by the data. However, as discussed before, we also have access to past observations of electricity demand up to two days prior to the forecasting time. Moreover, the NN benchmarks incorporate lagged target values when generating forecasts. Including the same information is therefore necessary to ensure a fair comparison. 

To account for the availability of past demand values, we introduce a time index $I$, denoting the most recent observation time of the target variable. In our setting, $I = 2$ days. We extend the previous model, by incorporating the signature of the path $Y_{[t-D, t-I]}$ obtained by linear interpolation of the observations $Y_{t-D}, \ldots, Y_{t-I}$, resulting in a model now denoted $RidgeSig(X,Y)$. We do not consider the signature of a joint path containing both past demand and temperature. The two series are not temporally aligned: at time $t$, we have access to temperature (forecast) $X_t$, whereas the most recent demand observation available is $Y_{t-I}$. Considering the joint path $\left(X_t, Y_{t-I}\right)$ would not be informative, as the time lag would prevent the signature from capturing the relevant joint dynamics.

In summary, fitting $RidgeSig(X,Y)$ consists of fitting a Ridge regression: 
$$
\widehat{\theta} \in \arg\min_{\theta} 
\sum_{t=w}^{t_{\mathrm{train}}}
\left(
\Delta_D Y_t
-
\theta^\top
\begin{bmatrix}
S^{\le N}\!\left(X_{[t-w,t]}\right) &
S^{\le N}\!\left(Y_{[t-D,t-I]}\right)
\end{bmatrix}
\right)^2
+ \lambda \lVert \theta \rVert_2^2,
$$
$$
\widehat{Y}_t = Y_{t-D} + \widehat{\Delta_D Y_t},
\quad
\widehat{\Delta_D Y_t}
=
\widehat{\theta}^\top
\begin{bmatrix}
S^{\le N}\!\left(X_{[t-w,t]}\right) &
S^{\le N}\!\left(Y_{[t-D,t-I]}\right)
\end{bmatrix}.
$$

Regarding hyperparameters, as suggested by synthetic data, we set the window size to $w=9$ and the truncation order to $N=6$.

\begin{table}[ht]
\centering
\resizebox{0.7\textwidth}{!}{%
\begin{tabular}{lccc}
\toprule
\textbf{Model} & \multicolumn{2}{c}{\textbf{Covariates}} & \textbf{Forecasts} \\
\cmidrule(lr){2-3}
 & \textbf{Exogenous} & \textbf{Endogenous} &  \\
\midrule
RidgeSig(X) 
& $X_{t-7\text{d}}, \ldots, X_{t+2\text{d}}$ 
& -- 
& $Y_{t+2\text{d}}$ \\

RidgeSig(X,Y) 
& $X_{t-7\text{d}}, \ldots, X_{t+2\text{d}}$ 
& $Y_{t-5\text{d}}, \ldots, Y_t$ 
& $Y_{t+2\text{d}}$ \\

TimeXer  
& $X_{t-7\text{d}}, \ldots, X_t$ 
& $Y_{t-7\text{d}}, \ldots, Y_t$ 
& $Y_{t+30\text{min}}, \ldots, Y_{t+2\text{d}}$ \\

N-BEATS  
& $X_{t-7\text{d}}, \ldots, X_t$ 
& $Y_{t-7\text{d}}, \ldots, Y_t$ 
& $Y_{t+30\text{min}}, \ldots, Y_{t+2\text{d}}$ \\

TiDE     
& $X_{t-7\text{d}}, \ldots, X_{t+2\text{d}}$ 
& $Y_{t-7\text{d}}, \ldots, Y_t$ 
& $Y_{t+30\text{min}}, \ldots, Y_{t+2\text{d}}$\\

\bottomrule
\end{tabular}%
}
\caption{Information used by each model at time $t$ to obtain the forecast $Y_{t+2\text{d}}$.}
\label{table:model_inputs}
\end{table}

\paragraph{Results}

The performance of the aforementioned models is summarized in Table~\ref{table:errors_real_data}. Both \textit{RidgeSig} models clearly outperform the linear regression and GAM, which rely on handcrafted features. We note that the number of parameters reported for the GAM is its effective degrees of freedom. Moreover, \textit{RidgeSig(X,Y)} improves upon \textit{RidgeSig(X)}, indicating that incorporating signatures of past demand as predictors improves forecasting accuracy.
We further observe that \textit{RidgeSig(X,Y)} outperforms both \textit{TimeXer} and \textit{N-BEATS}, despite having several orders of magnitude fewer parameters. This advantage likely stems, at least in part, from their limited ability to incorporate future covariates.
While \textit{TiDE} achieves the best overall performance, the improvement over \textit{RidgeSig(X,Y)} comes at the cost of increasing the number of parameters from 240 to 1.8 million. This highlights the strong performance–complexity trade-off of the signature-based approach, which achieves competitive accuracy with substantially lower model complexity and training burden.

Furthermore, Figure~\ref{fig:experiment3} compares demand predictions from the \textit{RidgeSig(X,Y)} model with those of GAM and \textit{TiDE}, focusing on a representative summer and winter week in 2015. The \textit{RidgeSig(X,Y)} model closely follows the observed values and captures the main demand patterns, particularly during winter, when electricity demand is more sensitive to temperature fluctuations. 

Continuing this seasonal analysis, Table~\ref{tab:seasonal_metrics} reports forecasting accuracy on the 2015 test year, defining winter as October–March and summer as April–September. We also construct an aggregated model (\textit{Agg}) combining \textit{RidgeSig(X,Y)} and \textit{TiDE} through a weighted sum of predictions, with weights estimated via Ridge regression on the validation set. The performance gap between the two original models is smaller in winter, when demand is more strongly driven by temperature. The aggregated model achieves the best results, particularly in winter, suggesting that signature features encode additional information that is complementary to the neural network representation. This indicates that incorporating signature-based models within aggregation or expert frameworks may be a promising direction for future research.

\begin{table}[ht]
\centering
\resizebox{0.7\textwidth}{!}{%
    \begin{tabular}{lccccc}
    \toprule
    \textbf{Model} & \textbf{Parameters} & \textbf{RMSE (MW)} & \textbf{MAPE (\%)}\\
    \midrule
    $LR\left( T, T^2, \overline{T}, \overline{T}^2, Y_{t-7\text{ days}} \right)$ & 5 & 3\,714 & 5.3 \\
    GAM &  37 & 3\,423 & 4.6 \\
    RidgeSig(X) & 120 & 3\,150 & 4.4 \\
    RidgeSig(X,Y) & 240 &\cellcolor{lightgray} 2\,960 &\cellcolor{lightgray} 4.1\\
    TimeXer & 647 K & 3\,263 & 4.4\\
    N-BEATS  & 840 K & 4\,273 & 5.6 \\
    TiDE  & 1.8M & \textbf{2\,474} & \textbf{3.0} \\
    \bottomrule
    \end{tabular}%
    }
    \caption{Test error metrics on real demand.}
    \label{table:errors_real_data}
\end{table}

Finally, we note that, given the methodological nature of our work, the experiments above are intended to establish the soundness and competitiveness of the signature-based approach rather than to achieve operational accuracy for electricity consumption forecasting. Several avenues of research remain for improving performance in a deployed setting. First, calendar variables such as day of the week and holidays, should be included, as they strongly influence consumption. A second promising direction is to move from the present static estimation to an adaptive one, in which the regression coefficients are updated online to track the changes during volatile periods such as the COVID pandemic. We leave a thorough investigation of these extensions to future work, and stress that the present study is a point of departure, not an endpoint, for the practical deployment of signature-based forecasting models.

 \begin{figure}[pos=htbp]
        \centering
    \includegraphics[width=0.8\linewidth]{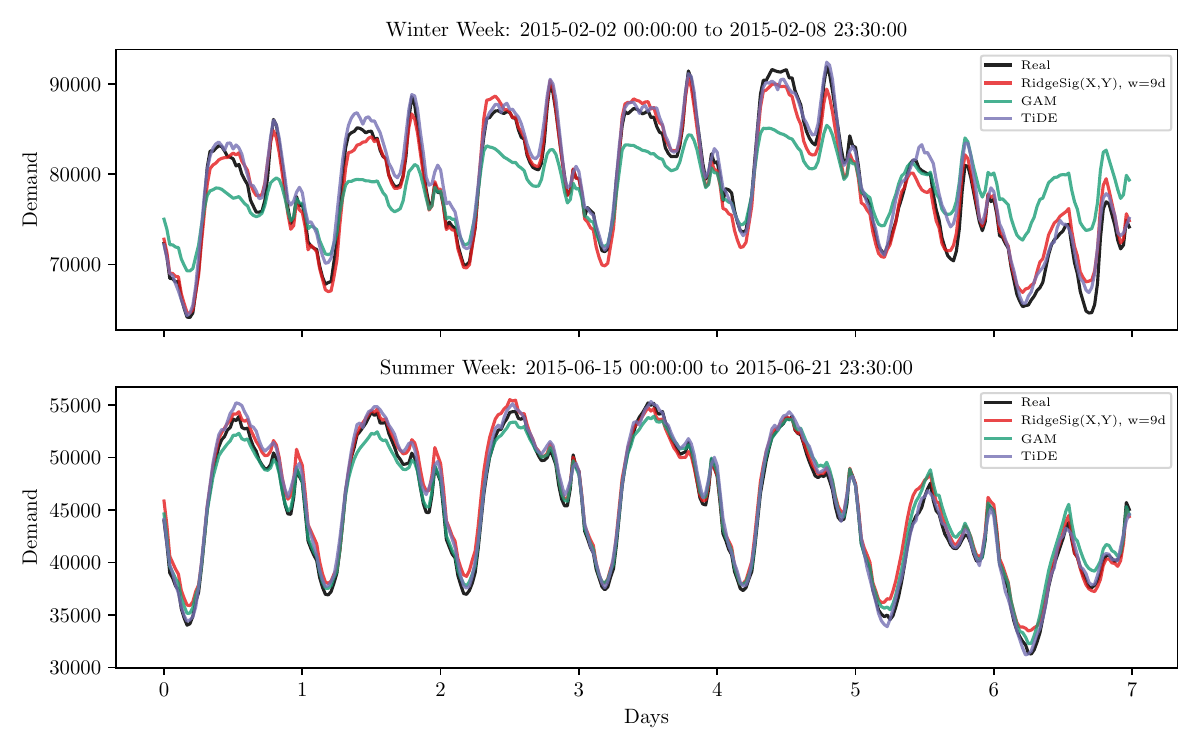}
        \caption{Observed and predicted demand for a winter (top) and summer (bottom) week in the test set.} 
        \label{fig:experiment3}
\end{figure} 

\begin{table}[ht]
\centering
\resizebox{0.7\textwidth}{!}{%
\begin{tabular}{lrrr|rrr}
\toprule
Season & \multicolumn{3}{c}{RMSE} & \multicolumn{3}{c}{MAPE (\%)} \\
\cmidrule(lr){2-4} \cmidrule(lr){5-7}
       & RidgeSig & TiDE & Agg & RidgeSig & TiDE & Agg \\
\midrule
Winter & 2\,990 & 2\,780.15 & 2\,529 & 3.52 & 2.99 & 2.78 \\
Summer & 2\,930 & 2\,130 & 2\,079 & 4.75 & 3.01 & 2.98 \\
\bottomrule
\end{tabular}%
}
\caption{Seasonal RMSE and MAPE for on the test year 2015.}
\label{tab:seasonal_metrics}
\end{table}

\subsection{Hyperparameter selection}
To conclude the experimental part, we briefly discuss the hyperparameters of the \textit{RidgeSig} model. We first note that the delay $D$ and the time of last information $I$ are both data-dependent. Tuning the look-back (sliding) window length is standard practice in comparable forecasting applications \citep{das2024longtermforecastingtidetimeseries, nie2023patchtst}. Indeed, each of the neural-network baselines considered in this paper likewise requires the window length to be chosen. If one wishes to avoid committing to a single window length, an alternative is to aggregate a family of experts indexed by different window sizes, as aggregating experts is already a well-established strategy in electricity load forecasting \citep{Gaillard2016AdditiveMA}. Another hyper-parameter is the ridge penalty $\lambda$, which may be tuned using the validation set. Finally, the truncation order $N$ governs model complexity and constitutes the principal computational limitation, since the dimension of the truncated signature grows exponentially with $N$. When the number of covariates is large, this growth can be mitigated either by partitioning them into several separate paths or by first reducing their dimension via PCA, as done in \citet{cohen2023nowcasting}. In summary, although our model involves several hyper-parameters, they are relatively few in number, interpretable, and have clear practical meaning; in contrast to neural-network architectures, whose performance often depends on numerous architectural choices, such as the depth and width of the network, the dropout rate, to name a few.

\section{Conclusion}\label{sec:conclusion}

In this work, we established theoretical foundations in a new framework for time series forecasting that relies on sliding-window signatures, proving universality of approximation and stationarity of signatures in the discrete-time setting. 
We highlighted the importance of aligning paths at a common starting point and including appropriate timestamps as an additional coordinate. Incorporating these insights, we built a model that leverages signatures as non-parametric features for covariate time series. Our sequential algorithm exploits sliding window overlap for efficient signature computation.
We demonstrated, both on synthetic and real electricity demand data, that signature features capture nonlinear and temporal dependencies without learned representations or expert knowledge, resulting in forecasts competitive with those based on hand-crafted features and those obtained from recent NN frameworks. Our approach does require sensible choices of hyperparameters, particularly window size and target delay, to match the data’s temporal structure and granularity. A limitation lies in the rapid growth of signature size with truncation order and path dimension, though we demonstrated good performance at modest orders (e.g., $N=6$).
Future directions include extending the theoretical results to a broader class of stochastic processes
and adapting our sequential algorithm to an online learning setting, better aligned with real-world applications.

\subsubsection*{Acknowledgements}

This research was supported by the European Union’s Horizon Europe research and innovation programme through the Marie Skłodowska-Curie Grant Agreement No.101081674.

N.D. would like to thank Prof. Christa Cuchiero for valuable inputs during meetings in Vienna, Adeline Fermanian for constructive discussions that helped improve this work, as well as Jeremy Francis Reizenstein for the modification of the \texttt{iisignature} package that made the sequential algorithm possible.

\bibliographystyle{elsarticle-harv}

\newpage

\appendix

\section{THEORETICAL BACKGROUND}
\label{appendix:theory}

In this section, we present a more rigorous treatment of results concerning paths and path integration, tensor spaces, and signatures. To ensure completeness and readability, some definitions and results from the main body are restated.

\subsection{Paths and path integration}

We define a \textit{continuous path} in $\mathbb{R}^d$ as any continuous mapping from some interval $[s,t]$ to $\mathbb{R}^d$, i.e. $x : [s, t] \longrightarrow \mathbb{R}^d$, $ u \longmapsto x_u = \big(x_u^{(1)}, x_u^{(2)}, \ldots, x_u^{(d)}\big)$. We consider multi-dimensional paths, meaning $d \geq 2$. Moreover, we restrict ourselves to working with paths of finite length, which we define through the notion of the $1$-variation.

\begin{definition*}
    Let $x:[s,t] \longrightarrow \mathbb{R}^d$ be a continuous path.
    The 1-variation of $x$ is defined by
    \begin{equation*}
        \|x\|_{\mathrm{1}-\mathrm{var}}=\sup _{\left(t_0, \ldots, t_k\right) \in P} \sum_{i=1}^k\left\|x_{t_i}-x_{t_{i-1}}\right\|,
    \end{equation*}
    where $P=\{(t_0, \ldots, t_k) \mid k \geq 0,$ $s=t_0<\cdots<t_k=t\}$
    denotes the set of all finite partitions of $[s, t]$.
\end{definition*}

Intuitively, $\|x\|_{\mathrm{1}-\mathrm{var}}$ can be interpreted as the length of the path $x$, therefore we work with paths such that $\|x\|_{\mathrm{1}-\mathrm{var}} < \infty$. These paths are also referred to as being of \textit{bounded variation}. This assumption will later prove important for theoretical guarantees, and it also allows us to define path integration in the $\textit{Riemann-Stieltjes}$ sense. We refer to Chapters 2 and 3 in \citet{friz2010multidimensional} for more details.

\begin{definition*}
    Let  $x, y$ be two continuous paths from $[s,t]$ to $\mathbb{R}$, with $x$ being of finite length. Let $P^n=\left\{s=t_0^n<t_1^n<\cdots<t_{N_n}^n=t\right\}$ for $n \geq 0$ be a sequence of partitions with vanishing mesh size $|P^n|=\max_{1 \leq i \leq N_n}|t_i - t_{i-1}|$, meaning $\left|P^n\right| \rightarrow 0$ as $n \rightarrow \infty$. The Riemann-Stieltjes integral of $y$ against $x$, is defined as
    \begin{equation*}
        \int_s^t y_u \mathrm{~d} x_u:=\lim _{n \rightarrow \infty} \sum_{i=0}^{N_n-1} y_{t_i^n}\left(x_{t_{i+1}^n}-x_{t_i^n}\right),
    \end{equation*}
    where the limit does not depend on the choice of the sequence of partitions $\left(P^n\right)_{n \geq 0}$.
\end{definition*}

This definition can be generalized to the case when $x$ and $y$ are vector-valued; for $x, y$ : $[s,t] \rightarrow \mathbb{R}^d$ we have
    \begin{equation*}
        \int_s^t y_u \mathrm{~d} x_u=\left(\begin{array}{c}
        \int_s^t y_u^{(1)} \mathrm{~d} x_u^{(1)} \\
        \vdots \\
        \int_s^t y_u^{(d)} \mathrm{~d} x_u^{(d)}
        \end{array}\right).
    \end{equation*}

\begin{rem}
    Take two continuous paths $x:[s,t] \longrightarrow \mathbb{R}$ and $y:[s,t] \longrightarrow \mathbb{R}$, where $x$ is continuously differentiable. Then the $\textit{Riemann-Stieltjes}$ integral of $y$ against $x$ comes down to 
    \begin{equation*}
        \int_s^t y_u \mathrm{~d} x_u=\int_s^t y_u {x}^{'}_u \mathrm{~d} u,
    \end{equation*}
    where the last integral is the classical Riemann integral and $ {x}^{'}_u = \mathrm{~d}x_u/\mathrm{~d}u$ denotes differentiation with respect to a single variable.
\end{rem}

\begin{example}
    Consider the constant path $y_u = 1$, $u \in [s,t]$. It follows that the integral of $y$ against any \newline $x:[s,t] \longrightarrow \mathbb{R}$ that is continuously differentiable is just the increment of $x$ on $[s,t]$:
    $$ 
    \int_s^t \mathrm{~d} x_u = \int_s^t {x}^{'}_u \mathrm{~d} u = x_t - x_s.
    $$
\end{example}

\subsection{Tensor spaces}

We denote by $\otimes$ the tensor product and by $(\mathbb{R}^d)^{\otimes k}$ the $k$-th tensor power of $\mathbb{R}^d$ , with $(\mathbb{R}^d)^{\otimes 0} := \mathbb{R}$. Let $e_1, \ldots, e_d$ be the canonical basis of $\mathbb{R}^d$, then $\left(e_{i_1} \otimes \cdots \otimes e_{i_k}\right)_{1 \leq i_1, \ldots, i_k \leq d}$ is an orthonormal basis of $\left(\mathbb{R}^d\right)^{\otimes k}$, meaning that any element $a \in\left(\mathbb{R}^d\right)^{\otimes k}$ can be written as $ a=\sum_{1 \leq i_1, \ldots, i_k \leq d} a^{\left(i_1, \ldots, i_k\right)} e_{i_1} \otimes \cdots \otimes e_{i_k}$, where $a^{\left(i_1, \ldots, i_k\right)} \in \mathbb{R}$.

 Furthermore, $\left(\mathbb{R}^d\right)^{\otimes k}$ is a Hilbert space of dimension $d^k$, with the following scalar product and norm:
$$
\langle a, b\rangle_{\left(\mathbb{R}^d\right)^{\otimes k}}=\sum_{1 \leq i_1, \ldots, i_k \leq d} a^{\left(i_1, \ldots, i_k\right)} b^{\left(i_1, \ldots, i_k\right)}, \quad \|a\|_{\left(\mathbb{R}^d\right)^{\otimes k}} =  \sqrt{\sum_{1 \leq i_1, \ldots, i_k \leq d} {\left( a^{\left(i_1, \ldots, i_k\right)} \right)} ^2} .
$$

\begin{definition}
    We define the extended tensor algebra $T\left(\left(\mathbb{R}^d\right)\right)$ as
    $$
    T\left(\left(\mathbb{R}^d\right)\right):=\left\{\left(a_0, a_1, \ldots, a_k,\ldots \right) \mid a_k \in\left(\mathbb{R}^d\right)^{\otimes k}, k \geq 0\right\}.
    $$
\end{definition}

In other words, the extended tensor algebra is a set of infinite sequences of tensors of increasing order: $a_0$ is a scalar, $a_1$ a vector, $a_2$ a matrix, $a_3$ a ``cube" (a third-order tensor), and so on. 
It can be shown that $T\left(\left(\mathbb{R}^d\right)\right)$ is a non-commutative algebra under the tensor product $\otimes$, with the neutral element $\mathbf{1}=(1,0,0,\ldots)$. Furthermore, it holds that the subset $T_1\left(\left(\mathbb{R}^n\right)\right):=\left\{\mathbf{a} \mid \mathbf{a} \in T\left(\left(\mathbb{R}^n\right)\right) \text { with } a_0=1\right\}$ is a Lie group. Finally, the following subspace will be of special interest. 

\begin{definition*}
    We denote with $\mathscr{T}\left(\mathbb{R}^d\right)$ the space of square-summable elements of $T\left(\left(\mathbb{R}^d\right)\right)$:
    
    \begin{equation*}
        \mathscr{T}\left(\mathbb{R}^d\right)=\left\{\mathbf{a} \in T\left(\left(\mathbb{R}^d\right)\right) \mid  \sum_{k=0}^{\infty}\left\|a_k\right\|_{\left(\mathbb{R}^d\right)^{\otimes k}}^2<\infty\right\}.
    \end{equation*}
        
\end{definition*}

We can  endow this space with the following scalar product and associated norm: for any $\mathbf{a}, \mathbf{b} \in \mathscr{T}\left(\mathbb{R}^d\right)$,
    \begin{equation}
        \langle \mathbf{a}, \mathbf{b}\rangle_{\mathscr{T}\left(\mathbb{R}^d\right)}=\sum_{k=0}^{\infty} \langle a_k, b_k\rangle_{\left(\mathbb{R}^d\right)^{\otimes k}}, \quad 
        \|\mathbf{a}\|_{\mathscr{T}\left(\mathbb{R}^d\right)} =  \sqrt{\sum_{k=0}^{\infty}\left\|a_k\right\|_{\left(\mathbb{R}^d\right)^{\otimes k}}^2} .
    \label{normhilbert}
    \end{equation}
    
\begin{proposition*}[Proposition \ref{prop:sig_Hilbert}]
     $\left(\mathscr{T}\left(\mathbb{R}^d\right),\langle\cdot, \cdot\rangle_{\mathscr{T}\left(\mathbb{R}^d\right)}\right)$ is a separable Hilbert space. 
\end{proposition*}

\begin{proof}
   (\textit{Completeness.}) The detailed proof that $\left(\mathscr{T}\left(\mathbb{R}^d\right), \langle\cdot, \cdot\rangle_{\mathscr{T}\left(\mathbb{R}^d\right)}\right)$ is a Hilbert space can be found in ~\citet{fermanian2021framingrnnkernelmethod}, Appendix A. 
   
    (\textit{Separability.}) Let $e_1, \ldots, e_d$ be the canonical basis of $\mathbb{R}^d$, and $B_k = \left\{ e_{i_1} \otimes \cdots \otimes e_{i_k} \mid 1 \leq i_1, \ldots, i_k \leq d \right\}$ denote the associated orthonormal basis of $\left(\mathbb{R}^d\right)^{\otimes k}$ that has $d^k$ elements. We now embed all of the elements of $B_k, k \geq 0$ in the space $\mathscr{T}\left(\mathbb{R}^d\right)$ by setting the other tensor sequence elements to zero, i.e. we define the sets
    \begin{equation*}
        \mathbf{B}_k = \left\{ \left(0, 0, \ldots, \overbrace{e_{i_1} \otimes \cdots \otimes e_{i_k} }^{k\text{-th position}}, 0, \ldots \right) \mid 1 \leq i_1, \ldots, i_k \leq d  \right\},
    \end{equation*}
     with $\mathbf{B}_0 = \left\{ \left( 1, 0, 0, \ldots \right)\right\}$. Let $\mathbf{B} = \bigcup_{k=0}^{\infty}\mathbf{B}_k$. The first thing to note is that $\mathbf{B}$ is a countable union of finite sets and is therefore countable. Secondly, $\mathbf{B}$ is a Schauder basis for $\mathscr{T}\left(\mathbb{R}^d\right)$, meaning that every element $\mathbf{a} \in \mathscr{T}\left(\mathbb{R}^d\right)$ can be uniquely represented as $ \mathbf{a}=\sum_{n=1}^{\infty} \alpha_n \mathbf{v}_n,$ $\mathbf{v}_n \in \mathbf{B}$, where the convergence of the infinite sum is the one of the topology induced by the norm \eqref{normhilbert}. Indeed, we have that $\mathbf{a}= \sum_{k=0}^{\infty} (0,0 \ldots, a_k, 0, \ldots)$ and $(0,0 \ldots, a_k, 0, \ldots) = \sum_{1 \leq i_1, \ldots, i_k \leq d} a^{\left(i_1, \ldots, i_k\right)} \left(0, 0, \ldots, e_{i_1} \otimes \cdots \otimes e_{i_k}, 0 \ldots \right) $, which gives us a unique representation of $\mathbf{a}$ in terms of elements of $\mathbf{B}$. Furthermore, from the fact that $B_k$ is an orthonormal basis for ${\left(\mathbb{R}^d\right)}^{\otimes k}$ and the definition of the scalar product \eqref{normhilbert}, it is easy to see that $\mathbf{B}$ is also orthonormal, making it a countable orthonormal basis. Finally, a Hilbert space is separable if and only if it has a countable orthonormal basis, which leads to the conclusion that $\left(\mathscr{T}\left(\mathbb{R}^d\right),\langle\cdot, \cdot\rangle_{\mathscr{T}\left(\mathbb{R}^d\right)}\right)$ is a separable Hilbert space.
\end{proof}

\subsection{Signatures}
\label{appendix:theory_signatures}

We can now fully define signatures, in the light of the theory presented earlier.

\subsubsection{Definitions}

\begin{definition*}
    The signature of a finite-length path $x$ on $[s, t]$ is defined as an infinite tensor sequence:
    \begin{equation*}
        S\left( x_{[s,t]} \right)=\left(1, \ S^1\left( x_{[s,t]} \right), \ldots, \ S^k\left( x_{[s,t]} \right), \ \ldots \right) \in T\left(\left(\mathbb{R}^d\right)\right),
    \end{equation*}
    where the $k$-th element (called level) is given by
    \begin{equation*}
        {S^k \left( x_{[s,t]} \right)}=\idotsint \displaylimits_{s < u_1<\cdots<u_k < t}  \mathrm{d} x_{u_1} \otimes \cdots \otimes \mathrm{d} x_{u_k} \in {\left(\mathbb{R}^d\right)}^{\otimes k}.
    \end{equation*}
\end{definition*}

Using the orthonormal basis for ${\left(\mathbb{R}^d\right)}^{\otimes k}$, we can write
    \begin{equation*}
       {S^k \left( x_{[s,t]} \right)} = \sum_{\left(i_1, \ldots, i_k\right) \subset\{1, \ldots, d\}^k} S^{\left(i_1, \ldots, i_k\right)}\left(x_{[s,t]}\right) e_{i_1} \otimes \cdots \otimes e_{i_k},
    \end{equation*}
    where $S^{(i_1, \ldots, i_k)}(x_{[s,t]})$ is referred to as the \textit{signature coefficient} of $x$ along the multi-index $(i_1, \ldots, i_k) \in \{1, \ldots, d\}^k$, $k \geq 1$, on $[s,t]$ and is given by
    \begin{equation*}
        S^{\left(i_1, \ldots, i_k\right)} = \idotsint \displaylimits_{s < u_1<\cdots<u_k < t} \mathrm{d}x^{\left(i_1\right)}_{u_1} \cdots \mathrm{d}x^{\left(i_k\right)}_{u_k}.
    \end{equation*}
    
In practice, we don't work with infinite sequences of tensors. Firstly, we only consider a finite number of signature levels, defining the \textit{truncated signature} as $S^{\leq N}\left( x_{[s,t]} \right) = \left(1, \ S^1\left( x_{[s,t]} \right), \ \ldots, \ S^N\left( x_{[s,t]} \right) \right)$,
where we refer to $N$ as the truncation order. Secondly, in applications, we disregard the underlying tensor structure, and consider the truncated signature as a collection of all signature coefficients with multi-index of length $k \leq N$, arranged in a vector 

\[
\scalebox{0.98}{$
\left(
1,
S^{(1)}(x_{[s,t]}),
\ldots,
S^{(d)}(x_{[s,t]}),
S^{(1,1)}(x_{[s,t]}),
S^{(1,2)}(x_{[s,t]}),
\ldots,
S^{\overbrace{(d,\ldots,d)}^{N\text{ times}}}(x_{[s,t]})
\right)
$}
\]
that is of size $s_d(N)= \sum_{k=0}^{N} d^k = \left(d^{N+1}-1\right) /(d-1)$.

We now illustrate the definitions on the example of a linear path, deriving a closed formula for signature coefficients in this simple case. This will prove an important building block in our methodology.

\begin{example}[Example \ref{ex-lin-sig} revisited]
    Let $x:[s,t] \mapsto \mathbb{R}^d$, $u \to \frac{x_{t} - x_s}{t-s}(u-s) + x_{s}$ be a d-dimensional linear path. For any $\left(i_1, \ldots, i_k\right) \in\{1, \ldots, d\}^k$ we have the following:
    
    \begin{align*}
    & S^{\left(i_1, \ldots, i_k\right)} \left(x_{[s,t]}\right)  = \idotsint \displaylimits_{s < u_1<\cdots<u_k < t} \mathrm{~d} x_{u_1}^{(i_1)} \cdots \mathrm{~d} x_{u_k}^{(i_k)} \\ 
    &= \idotsint \displaylimits_{s < u_1<\cdots<u_k < t} \frac{x_t^{(i_1)}-x_s^{(i_1)}}{t-s} \mathrm{~d} u_1 \ldots \frac{x_t^{(i_k)}-x_s^{(i_k)}}{t-s}  \mathrm{~d} u_k  \\
    &=\prod_{j=1}^k \left(\frac{x_t^{(i_j)}-x_s^{(i_j)}}{t-s}\right) \cdot \idotsint \displaylimits_{s < u_1<\cdots<u_k < t} \mathrm{~d} u_1  \ldots \mathrm{~d} u_k \\
    &= \prod_{j=1}^k \left(\frac{x_t^{(i_j)}-x_s^{(i_j)}}{t-s}\right) \cdot \int_s^t \int_s^{u_k} \ldots \left( \int_s^{u_2} \mathrm{~d} u_1  \right) \ldots \mathrm{~d} u_{k-1} \mathrm{~d} u_k  \\
    &= \prod_{j=1}^k \left(\frac{x_t^{(i_j)}-x_s^{(i_j)}}{t-s}\right) \cdot \int_s^t \int_s^{u_k} \ldots \left( \int_s^{u_3} \left(u_2 - s\right) \mathrm{~d} u_2  \right) \ldots \mathrm{~d} u_{k-1} \mathrm{~d} u_k  \\
    &= \prod_{j=1}^k \left(\frac{x_t^{(i_j)}-x_s^{(i_j)}}{t-s}\right) \cdot \int_s^t \int_s^{u_k} \ldots \left( \int_s^{u_4} \frac{{\left(u_3 - s\right)}^2}{2} \mathrm{~d} u_3  \right) \!\!\!\ldots \mathrm{~d} u_{k-1} \mathrm{~d} u_k \\
    &= \ldots =  \prod_{j=1}^k \left(\frac{x_t^{(i_j)}-x_s^{(i_j)}}{t-s}\right) \cdot \frac{{(t-s)}^k}{k!} 
    =\frac{1}{k!} \prod_{j=1}^k\left(x_t^{(i_j)}-x_s^{(i_j)}\right) .
    \label{siglin}
    \end{align*}
    
    More compactly written, $S^k\left( x_{[s,t]} \right)=\frac{1}{k!}\left(x_t-x_s\right)^{\otimes k}$.
    \label{ex:piecewise_2}
\end{example}

\subsubsection{Properties}

\begin{proposition}[Invariances]\label{prop:invariance} Let $x:[s,t] \mapsto \mathbb{R}^d$ be a path of finite length. The following holds:
\begin{enumerate}
    \item Let $\widetilde{x}_u=x_{\psi(u)}$ be the smooth reparametrization of $x$, where $\psi:[s,t] \rightarrow[s,t]$ is a continuously differentiable non-decreasing surjection. Then, for any $[v,w] \subset [s,t]$ we have $S\left(\widetilde{x}_{[v,w]} \right)=S\left(x_{\left[\psi(v), \psi(w)\right]}\right) .$
    \item  Let $\overline{x}_u = x_u + a$ be a path obtained by translating $x$ by some $a \in \mathbb{R}^d$. Then 
    $
    S\left(\overline{x}_{[s,t]}\right) = S\left( x_{[s,t]} \right).
    $
\end{enumerate}
\label{prop:invariances}
\end{proposition}

Both of these properties can be derived straight from the definition of signatures as iterated integrals. Invariance to reparametrization is a consequence of the change of variable formula in Riemann-Stieltjes integration, while invariance to translation follows from the fact that $d\overline{x}_t = dx_t $.

\begin{proposition*}[Algebraic properties]
Let $x:[s, t] \mapsto \mathbb{R}^d$ and $y:[t, u] \mapsto \mathbb{R}^d$ denote two paths of finite length.
\begin{enumerate}
    
    \item (\textbf{Chen's identity}) Let $x * y:[s, u] \mapsto \mathbb{R}^d$ be the concatenation of $x$ and $y$, meaning $(x * y)_v=x_v$ for $v \in[s, t]$ and $(x * y)_v=x_t+y_v-y_t$ for $v \in[t, u]$. Then
    $
    S(\left(x * y\right)_{[s,u]})=S\left( x_{[s,t]} \right) \otimes S(y_{[t,u]})
    $
    
    \item (\textbf{Time reversal}) We denote the time-reversal of $x$ as the path $\overleftarrow{x}:[s, t] \mapsto \mathbb{R}^d$ where $\overleftarrow{x}(u)=x_{s+t-u}$. Then
    $
    S\left( x_{[s,t]} \right) \otimes S\left( \overleftarrow{x}_{[s,t]} \right) = (1,0,0,\ldots).
    $
\end{enumerate}
\end{proposition*}

We can rephrase Chen's relation in terms of signature coefficients: for any multi-index $\left(i_1, \ldots, i_k\right) \subset\{1, \ldots, d\}^k$, it holds that
\begin{equation}
    S^{(i_1, \ldots, i_k)}\left((x * y)_{[s,u]}\right)
    = \sum_{\ell=0}^k S^{(i_1, \ldots, i_\ell)}\left(x_{[s,t]}\right) \cdot S^{(i_{\ell+1}, \ldots, i_k)}\left(y_{[s,t]}\right)  \label{chencomp}.
\end{equation}

This result is crucial for the operationalization of signature computation as it provides a recursive formula for the signature of a concatenation of paths. As noted earlier, it also enables an easy calculation in the case of piecewise linear paths, which we explain in greater detail in the following remark.

\begin{rem}

\label{rem:piecewiselin_calc}

    Suppose that $x$ is a piecewise linear path such that for $s = u_1 < u_2 \ldots < u_k = t$, $x|_{[u_{\ell}, u_{\ell+1}]}$, $1 \leq \ell \leq k$, is linear. In order to obtain its signature, we first need to calculate the signature coefficients for each linear segment, as in Example \ref{ex:piecewise_2}:
    \begin{equation*}
        {S^{(i_1, \ldots, i_k)}\left(x_{[u_{\ell},u_{\ell+1}]} \right)} = \frac{1}{k!} \prod_{j=1}^k\left(x_{u_{\ell+1}}^{(i_j)}-x_{u_{\ell}}^{(i_j)}\right).
        \label{siglin1}
    \end{equation*}

    We proceed to calculate the signature on the whole time horizon $[s,t]$ by inductively concatenating the linear segments using the Chen's relation \eqref{chencomp}:
    
\begin{align} \label{siglin2}
S^{(i_1, \ldots, i_k)} & \left(x_{[s,u_{\ell+1}]} \right) 
= \sum_{m=0}^k S^{(i_1, \ldots, i_m)}\left(x_{[s,u_\ell]} \right) 
   \cdot S^{(i_{m+1}, \ldots, i_k)}\left(x_{[u_{\ell},u_{\ell+1}]} \right) \nonumber \\
&= \sum_{m=0}^k \left( S^{(i_1, \ldots, i_m)}\left(x_{[s,u_{\ell}]} \right)  
   \cdot \frac{1}{(k - m)!} \prod_{j=m + 1}^{k} \left(x_{u_{\ell+1}}^{(i_j)} - x_{u_\ell}^{(i_j)}\right) \right) .
\end{align}

\end{rem}

Computing the signature of a piecewise linear path therefore reduces to applying the iterative formula stated above, requiring no numerical integration. This procedure is efficiently implemented in the Python package \texttt{iisignature} ~\citep{iisignature}.

Going back to the time reversal property, we remark that, by definition, the first element of the signature is set to 1 and therefore $S\left( x_{[s,t]} \right) \in T_1\left(\left(\mathbb{R}^d\right)\right)$. This result now states that the inverse of the signature of a path $x$ in the tensor group $T_1\left(\left(\mathbb{R}^d\right)\right)$ is exactly the signature of $x$ traversed backwards in time.

\section{PROOFS OF THEORETICAL CONTRIBUTIONS}

\subsection{Proof of Theorem \ref{thm:uat-signatures}}
\label{appendix:theory_uat}

\begin{theorem*}[Universal approximation on sliding windows]
    Let ${(X_t)}_t$ be a $d$-dimensional time series such that Assumption \ref{assumption:bounded} holds. Let $w \in \mathbb{N}$ be a fixed window size and $f : C\left([0,w],\mathbb{R}^{d+1}\right) \mapsto \mathbb{R}$ a continuous function given the uniform topology. 
    It holds that for every $\varepsilon>0$, there exists $N \in \mathbb{N}, \theta \in \mathbb{R}^{s_{d+1}(N)}$, such that
    \begin{equation*}
        \sup_{t}\left\lvert f\left({X}_{[t-w, t]} - (0,X_{t-w})\right)- \theta^{\top} S^{\leq N} \left({X}_{[t-w, t]}\right) \right\rvert \leq \varepsilon.
    \end{equation*}
\end{theorem*}

\begin{proof}
 We first introduce $\overline{X}_{[t-w, t]} := {X}_{[t-w, t]} - (0,X_{t-w})$ as a random path obtained by linearly interpolating $\left(0, 0\right), \left(1, X_{t-w+1} - X_{t-w}\right), \ldots, \left( w, X_{t} - X_{t-1} \right) $. In other words, we consider paths on windows translated by the starting point. We now define $U = \left\{ \overline{X}_{[t-w, t]}(\omega) \mid \omega \in \Omega; \hspace{2mm} t \in \mathbb{Z}\right\}$ as the set of all realisations of random paths on all sliding windows.
 We note that $U$ is a set of piecewise linear functions whose first coordinate (corresponding to time) is strictly monotone and such that $\left| \right(\overline{X}^{(i)}_{[t-w, t]}(\omega) \left)\right| \leq \max(2M,w) =:B$ for all $1 \leq i \leq d$. It follows that they are therefore Lipschitz-continuous with the same Lipschitz constant $L := \max(4M, 1)$. 

\vspace{1.5mm}

 Let us define a set of paths $K \subset C\left([0,w],\mathbb{R}^{d+1}\right)$ such that: 
\begin{enumerate}
    \item $ \forall x \in K$ has at least one coordinate that is linear, i.e. $\exists j \in \{1,2,\ldots, d\}$, $x^j_u = a_xu$ and such that there exists $\varepsilon > 0$ such that $a_x \geq \varepsilon$ $\forall x \in K$ \label{linear-monotone},
    \item $\forall x\in K$ has the same starting point at 0, i.e. $x_0 = 0\in \mathbb{R}^{d+1}$ ,
    \item the paths are uniformly bounded by $B$, i.e. $\left|x^{(i)}_u \right| \leq B$ $\forall i \in \{1,2,\ldots,d\}$, $\forall x \in K$ ,
    \item the paths are Lipschitz with the same constant $L$, i.e.  $\| x_v -x_u\|_{\infty} \leq L |v-u| $, $\forall x \in K.$
\end{enumerate}

Functions in $K$ have the same Lipschitz constant, therefore they are also uniformly equicontinuous. 
Furthermore, it can be shown that the set $K$ is also closed in $C\left([0,w], \mathbb{R}^{d+1}\right)$. We can now use the Arzelà–Ascoli theorem to conclude that $K$ is compact in $C\left([0,w], \mathbb{R}^{d+1}\right)$ in the uniform topology. This allows us to apply the Theorem \ref{thm:uat} to $K$.

Leveraging Theorem \ref{thm:uat} and the fact that $U \subset K$, we obtain:
\begin{equation*}
 \sup_{U}\Bigg|f\left(\overline{X}_{[t-w, t]}(\omega)\right)- \theta^{\top} S^{\leq N} \left(\overline{X}_{[t-w, t]}\right) \Bigg| 
\leq 
\sup_{K} \bigg| f\left({x}\right) - \theta^{\top} S^{\leq N} \left({x}_{[0,w]}\right) \bigg| \leq \varepsilon .
\end{equation*}
We conclude using the identity $S^{\leq N} \left(\overline{X}_{[t-w, t]}\right)=S^{\leq N} \left({X}_{[t-w, t]}\right)$ following from the invariance by translation of the signatures; see Proposition \ref{prop:invariance}, 2.
\end{proof}

\subsection{Proof of Theorem \ref{thm:stationarity_sig}}
\label{appendix:theory_stat}

We present a more formal statement of the result with the corresponding proof.

\begin{theorem*}
Let ${(X_t)}_{t \ge 1}$ be a discrete $d$-dimensional time series such that the Assumption \ref{assumption:stationary} is satisfied. The following statements hold:
\begin{enumerate}[label=\roman*)]
        \item For any multi-index $\left(i_1, \ldots, i_k\right)$ of any length $k \in \mathbb{N}$, the time series ${\left({S^{(i_1, \ldots, i_k)}(X_{[t-w,t]})}\right)}_t$ is strictly stationary. 
        
        \item For any truncation order $N \in \mathbb{N}$, the time series ${\left({S^{\leq N}\left(X_{[t-w,t]}\right)}\right)}_t$ is strictly stationary.

        \item The time series ${\left({S\left(X_{[t-w,t]}\right)}\right)}_t$ is strictly stationary.
\end{enumerate}
\end{theorem*}

\begin{proof}
\begin{enumerate}[label=\roman*)]
    \item  For any multi-index $I=(i_1, \ldots, i_k)$ of length $k \in \mathbb{N}$, we denote with $f^{I}$ the function from $\mathbb{R}^{w \times d}$ to $\mathbb{R}$ given as:
    \begin{equation*}
        (X_{t-w+1} - X_{t-w},\ldots, X_{t} - X_{t-1})  
        \overset{f^I}{\mapsto} {S^{I}(X_{[t-w,t]})} \,.
    \end{equation*}
  
    Examining the iterative formula \eqref{siglin2}, we see that the $f^I$ (the signature component ${{S^{I}(X_{[t-w,t]})}}$) is the same deterministic, measurable function (as a composition of measurable functions: multiplication, sums and powers) of the increments $(X_{t-w+1} - X_{t-w}, X_{t-w+2} - X_{t-w+1}, \ldots, X_{t} - X_{t-1})$ of the original time series ${(X_t)}_{t}$. Furthermore, it does not depend on $t$. We also note that if the $d$-dimensional time series $\left( X_t - X_{t-1}\right)_t$ is strictly stationary, then it follows straight from the definition of strict stationarity that the $(d \times w)$-dimensional time series ${(X_{t-w+1} - X_{t-w},}$ ${\ldots, X_{t} - X_{t-1})}_t$ is also strictly stationary. If we now remember that for any deterministic, measurable function $g$ and any random vectors $\mathbf{X}, \mathbf{Y}$ such that $\mathbf{X} \overset{D}{=} \mathbf{Y}$, we have $g(\mathbf{X}) \overset{D}{=} g(\mathbf{Y}) $, and combine it with the previous remark for $f^I$, we have the strict stationarity for ${\left({S^{I}(X_{[t-w,t]})}\right)}_t$.
    
    \item Let us now fix an $N$ and consider the $s_d(N)$-dimensional time series of the signature truncated at level $N$, ${\left({S^{\leq N}\left(X_{[t-w,t]}\right)}\right)}_t$. Let $I_1, I_2, \! \ldots, I_{s_d(N)}$ be the multi-indexes of the components of the truncated signature, in alphabetical order. We now have:
    \begin{align*}
        & f^{\leq N}  \left( X_{t-w+1} - X_{t-w},  \ldots, X_{t} - X_{t-1} \right)  \\
        & := 
        \left( f^{I_1}, \ldots, f^{I_{s_d(N)}} \right) 
        ( X_{t-w+1} - X_{t-w}, \ldots, X_{t} - X_{t-1} ) \\
        &= \left( {S^{I_1}\left(X_{[t-w,t]}\right)}, \ldots, S^{I_{s_d(N)}}\left(X_{[t-w,t]}\right) \right) \\
        &= {S^{\leq N}\left(X_{[t-w,t]}\right)} ,
    \end{align*}
where $f = \left( f^{I_1}, \ldots, f^{I_{s_d(N)}} \right) : \mathbb{R}^{d \times w} \longrightarrow \mathbb{R}^{s_d(N)}$ is a deterministic, measurable function of the increments that does not depend on $t$, as noted above. Following the same arguments as before, we can now conclude that ${\left({S^{\leq N}\left(X_{[t-w,t]}\right)}\right)}_t$ is strictly stationary. 

\item It is easy to see that the previous proof holds not only for the time series of signatures truncated at level $N$, but also for the time series of any collection of signature components denoted by multi-indexes $\left( I_1, \ldots, I_j \right)$. Additionally, as discussed in Remark \ref{rem-stat_def}, the finite-dimensional distributions of the signature components uniquely determine the law of the entire infinite-dimensional time series of signatures. From these two arguments, it follows that ${\left({S \left(X_{[t-w,t]}\right)}\right)}_t$ is strictly stationary.

\end{enumerate}

\end{proof}

\section{IMPLEMENTATION DETAILS}

\label{appendix:implementation}

Our method consists in efficiently computing signatures on sliding windows and fitting a Ridge regression model, using the covariate signatures as predictors and the increments of the target series as the response variable. Hyperparameters are tuned on a validation set. The full procedure is detailed in the pseudocode below.

\begin{algorithm}[h!]
    \SetKwInOut{Input}{Input}
    \SetKwInOut{Output}{Output}
    \Input{data $\left\{ (x_0,y_0), \ldots, (x_n,y_n) \right\}$, truncation order $N$, delay $D$, window size $w$} 
    Split the data into the train and validation and test set by fixing the horizons $w \leq  t_{train} < t_{valid} \leq n$ \;
    \For {$\lambda$ in \textit{set of possible regularization parameters}}{
        Compute the signatures $S\left( x_{[t-w,t]}\right)$ for $w \leq t\leq t_{train}$ using Algorithm \ref{algo:update_sig} \;
        Fit a Ridge regression model with parameter $\lambda$ on pairs $\left\{ \left( S\left( x_{[t-w,t]}\right) , \Delta_D y_t \right) \mid w \leq t\leq t_{train} \right\}$ \;
        Compute forecasts $\widehat{Y_t}=y_{t-D}+{\widehat{\theta}}^{\top}S\left( x_{[t-w,t]}\right)$ and error metric on the validation set \;
    }
    Identify the best regularisation parameter $\widehat{\lambda}$  (smallest error metric on the validation set) \;
    Refit the Ridge regression with parameter $\widehat{\lambda}$ on the combined train and validation set \;
	\Output{Ridge coefficients estimator $\widehat{\theta}$}
	\caption{Fitting a Ridge regression on signatures}
	\label{algo:ridge_sig}
\end{algorithm}

Finally, to obtain the forecast $\widehat{Y}_t$ for $t$ outside the train-validation set, it is enough to update the signature as in Algorithm \ref{algo:update_sig} and use the estimated Ridge regression coefficients $\widehat{\theta}$ to compute the forecast $\widehat{Y}_t = y_{t-D} + \widehat{\theta}^\top S(x_{[t-{w},t]})$. 

We note that this sequential procedure can easily be transferred to an online learning setting.

\section{EXPERIMENTAL DETAILS}
\label{appendix:experiments}

\subsection{Technical details} 
The data and code used in this paper are publicly available in our GitHub repository: \url{https://github.com/ninadrobac/slidesig}.
All experiments were conducted on a standard laptop (MacBook Air, Apple M2 chip, 16 GB RAM), demonstrating that our implementation can be executed efficiently on standard hardware. 


\subsection{Data}

\paragraph{Real data} We use half-hourly electricity demand data from RTE (France’s electricity transmission system operator), aggregated on the national level, and temperature observations from Météo France recorded on a three-hour grid. The temperature series on the national level is obtained by averaging across weather stations and linearly interpolating to match the demand frequency. This yields a unified dataset with 48 observations per day from January 1, 2012, to December 30, 2015, resulting in 70128 data points of temperature and demand. This period is selected to avoid the demand fluctuations associated with the COVID-19 pandemic and the 2022 energy crisis. We denote by $Y_t$ the electricity demand and by $T_t$ the temperature at any time $t \geq 1$.

\paragraph{Synthetic data} Electricity consumption depends on many factors, including calendar variables such as weekdays and public holidays. Since our goal is to isolate the role of temperature, we construct a synthetic dataset in which both the strength of past dependency and the functional form of the relationship can be controlled. To model memory effects, we introduce exponentially smoothed temperature as
\begin{equation*}
    \left\{
    \begin{array}{ll}
        \overline{T}_1^\alpha = T_1 \\
      \overline{T}_t^\alpha = (1-\alpha)\overline{T}_{t-1}^\alpha + \alpha T_t^\alpha \,, \,\textrm{ for any } t \geq 2  \,, 
    \end{array}
\right.
\end{equation*}
where the smoothing parameter $\alpha \in (0,1]$ controls the influence of past values - the smaller $\alpha$, the stronger the influence of past values. We then fit a linear regression model with observed consumption $Y_t$ as the target variable and $\overline{T}_t$ and $\overline{T}_t^2$ as covariates:
\begin{equation}
    \widehat{Y}^\alpha_t = \widehat{\theta_1} + \widehat{\theta_2} \overline{T}_t^\alpha +  \widehat{\theta_3} \big(\overline{T}_t^\alpha\big)^2,
    \label{eq:lin_model_synthetic}
\end{equation}

thereby imposing a quadratic dependence of simulated demand on smoothed temperature. The final synthetic demand time series $\big(\widetilde{Y}^\alpha_t\big)_t$ is obtained by adding normally distributed noise to the fitted values $\widehat{Y}^\alpha_t$ :
\begin{equation}
    \widetilde{Y}^\alpha_t = \widehat{\theta_1} + \widehat{\theta_2} \overline{T}_t^\alpha +  \widehat{\theta_3} \big(\overline{T}_t^\alpha\big)^2 + \varepsilon_t, \quad \varepsilon_t \overset{\mathrm{i.i.d}}{\sim} \mathcal{N}(0,\sigma^2)\,.
    \label{eq:model_synthetic}
\end{equation}

In all experiments presented, we fix $\sigma$ to introduce a moderate level of noise ($\sigma=1000$) and we set a fixed random seed when simulating the noise to ensure reproducibility. The synthetic demand series closely follows the main patterns of the observed series, effectively capturing the yearly variations driven by temperature, providing a realistic proxy for evaluation. However, as illustrated in Figure~\ref{fig:consom_sim_real_comparison}, Model~\eqref{eq:model_synthetic} does not reproduce finer weekly patterns that depend on calendar variables such as the day of the week or time of day, as it only contains temperature information.
This is reflected in the fact that for experiments on synthetic data we fix $D=2$, and for experiments on real data we choose $D=7$ to account for weekly patterns. 
Although we could have adapted Model~\eqref{eq:model_synthetic} by adding daily seasonality to better match the data locally, this is simply a proof of concept, so we opted for the simplest possible framework.

We note that by varying the smoothing parameter $\alpha$, we can generate different versions of the dataset in which past values exert more or less influence on the simulated demand. Over several tested parameters, we mainly focus on $\alpha = 0.005$ as it offers the best fit of Model~\eqref{eq:model_synthetic} to the real data.

\begin{figure}[pos=htbp]
    \centering
    \begin{minipage}[b]{0.48\linewidth}
        \centering
        \includegraphics[width=\linewidth]{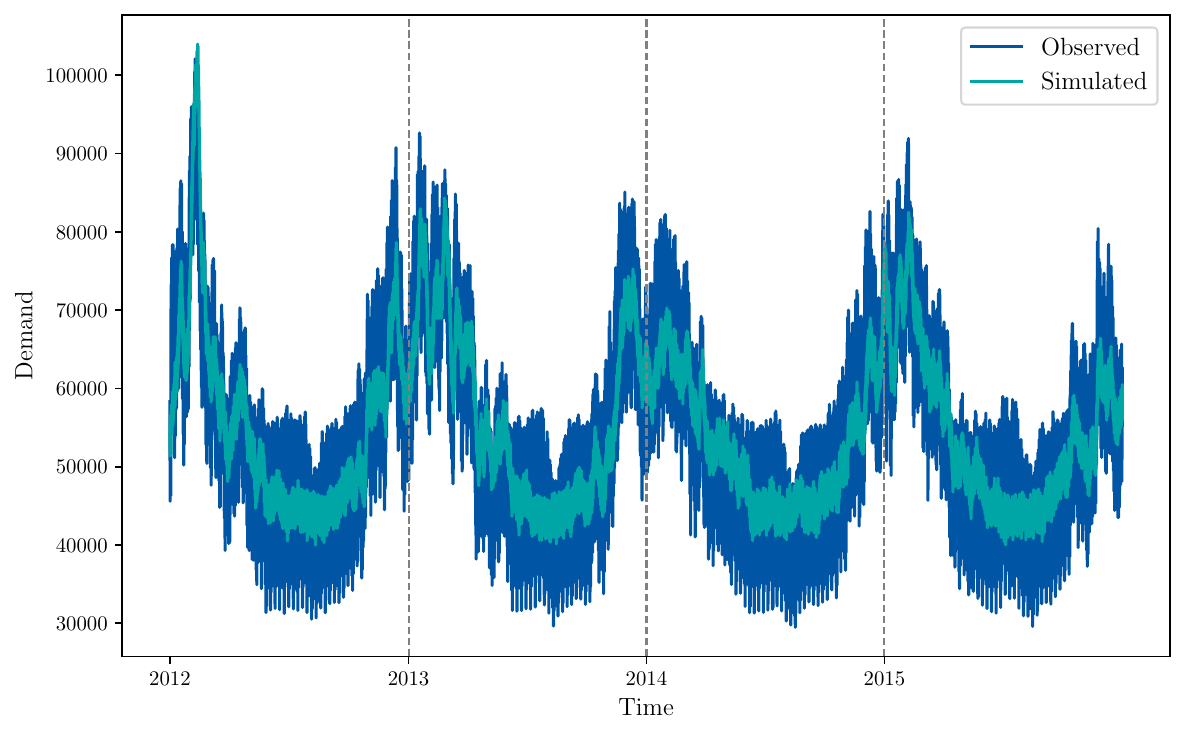}
    \end{minipage}
    \hfill
    \begin{minipage}[b]{0.48\linewidth}
        \centering
        \includegraphics[width=\linewidth]{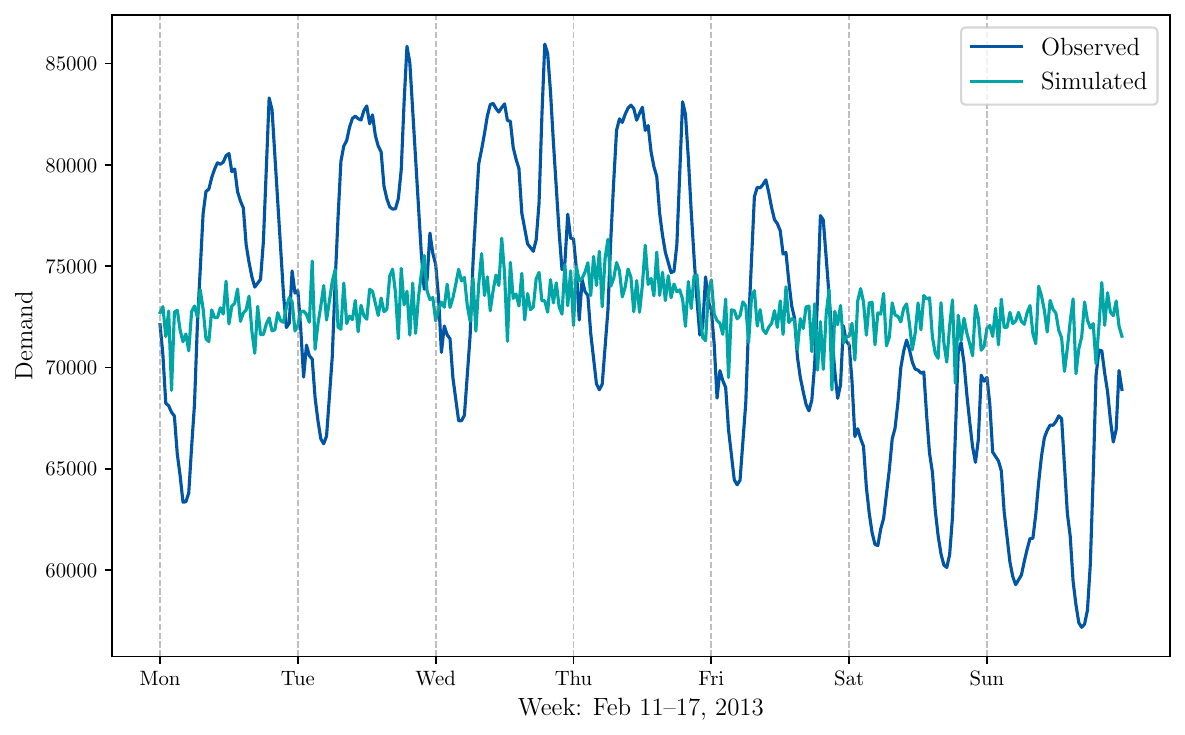}
    \end{minipage}
    \caption{Comparison of observed and simulated electrical demand. Left: full four-year dataset showing that the synthetic series captures yearly patterns driven by temperature. Right: zoom-in on a weekly segment (Feb 11--17, 2013) illustrating that weekly fluctuations driven by calendar variables are not retained in synthetic demand.}
    \label{fig:consom_sim_real_comparison}
\end{figure}

\paragraph{Assumptions in practice}
As discussed earlier, Assumption~\ref{assumption:bounded} is satisfied since temperature values in the dataset remain within natural bounds (between $-10\,^{\circ}\mathrm{C}$ and $35\,^{\circ}\mathrm{C}$). In contrast, verifying Assumption~\ref{assumption:stationary} is less straightforward, as it concerns the distributional properties of temperature increments. To assess this, we apply two standard statistical tests — the \textit{Augmented Dickey–Fuller} (ADF) and \textit{Kwiatkowski–Phillips–Schmidt–Shin} (KPSS) tests, both of which indicate that the time series of half-hourly temperature increments is weakly stationary at a significance level of $0.05$. The same analysis was performed for weekly demand increments and two-day increments of simulated demand, yielding satisfactory results.

\subsection{Implementation of NN models}
\label{appendix:implementation_nn}
For the \textit{N-BEATS} and \textit{TiDE} architectures, we relied on the implementation provided by the Python \texttt{darts} library \citep{JMLR:v23:21-1177}. We have also considered including the \textit{NBEATSx} model \citep{Olivares_2023} from the \texttt{NeuralForecast} package \citep{olivares2022library_neuralforecast}; however, this framework did not provide results beyond the first prediction horizon, making it impossible to obtain rolling forecasts comparable with the \textit{RidgeSig} models. We note that, although the original \textit{N-BEATSx} architecture supports the inclusion of future covariates, the available package version did not include this functionality, and adapting the GitHub code corresponding to the original paper proved challenging. The \textit{TimeXer} model was implemented by adapting the official code released by the authors on GitHub.

For all models, architectural hyperparameters, such as the number of layers and hidden dimensions, were selected based on recommendations and example configurations provided in the respective papers and accompanying repositories.

Regarding data-related hyperparameters, given the half-hourly sampling frequency and the weekly cycle observed in electricity demand, a look-back window of seven days (336 observations) was used as input to predict a forecast horizon of two days (96 observations). Input data were created through overlapping windows, analogous to the \textit{RidgeSig} models, with a stride parameter set to one. Temperature was included as an exogenous covariate and provided to the models alongside the electricity demand series.

The dataset was split chronologically into training (2012–2013), validation (2014), and test (2015) sets. Both the target variable and temperature covariate were standardized using \texttt{StandardScaler} fitted on the training data and applied to the validation and test sets. Training was performed for 20 epochs using mini-batch gradient descent. The \textit{TimeXer} custom implementation was built in PyTorch, with training managed through the PyTorch Lightning framework. Model performance was evaluated on the test set by generating rolling predictions, and forecast outputs were rescaled to the original units before computing evaluation metrics. To facilitate comparison with the point forecasts obtained with the signature-based models, only the last point of the forecasting window was taken as a 2-day-ahead prediction.

All experiments were conducted with fixed random seeds, and the full implementation has been made publicly available to ensure reproducibility.

\bibliography{references}

\end{document}